\documentclass[journal]{IEEEtran}
\usepackage{enumerate}

\usepackage{tabularx}
\usepackage{makecell}

\usepackage[table,xcdraw]{xcolor}
\usepackage{threeparttable}
\usepackage{algpseudocode}
\usepackage{setspace}

\usepackage{stfloats}
\usepackage{amsfonts}
\usepackage{amssymb}
\usepackage{amsthm}
\usepackage{cite}
\usepackage[cmex10]{amsmath}
\usepackage{float}
\usepackage{color}
\usepackage{stfloats,fancyhdr}
\usepackage{amsmath,bm}
\usepackage{algorithm}
\usepackage{multirow}
\usepackage{changepage}
\usepackage[normalem]{ulem}
\usepackage{balance}

\newtheorem{theorem}{Theorem}

\newtheorem{proposition}{Proposition}
\newtheorem{corollary}{Corollary}
\newtheorem{definition}{Definition}
\newtheorem{assumption}{Assumption}

\ifCLASSINFOpdf
\usepackage[pdftex]{graphicx}
\DeclareGraphicsExtensions{.pdf,.jpeg,.png}
\else
\usepackage[dvips]{graphicx}
\DeclareGraphicsExtensions{.eps}
\fi

\usepackage{subfigure}
\usepackage{fancybox,dashbox}
\usepackage{authblk}
\usepackage{tabularx}

\usepackage{mathtools}
\mathtoolsset{showonlyrefs=true}

\begin{document}

\title{\huge Wireless Human-Machine Collaboration in Industry 5.0}
%\title{Wireless Human-Machine Collaborative Control in Industry 5.0}

\author{Gaoyang Pang, Wanchun Liu*,~\IEEEmembership{Member,~IEEE,}
%Zhibo Pang,~\IEEEmembership{Senior Member,~IEEE,} 
Dusit Niyato,~\IEEEmembership{Fellow,~IEEE,}
Daniel E. Quevedo,~\IEEEmembership{Fellow,~IEEE},
Branka Vucetic,~\IEEEmembership{Life Fellow,~IEEE,} 
Yonghui Li,~\IEEEmembership{Fellow,~IEEE} 
\vspace{-0.5cm}
        % <-this % stops a space
\thanks{The work of W. Liu was supported by the Australian Research Council’s Discovery Early Career Researcher Award (DECRA) Project DE230100016. \textit{(Corresponding author: W. Liu.)}}
\thanks{G. Pang, W. Liu, D. Quevedo, B. Vucetic, and Y. Li are with the School of Electrical and Computer Engineering, The University of Sydney, Sydney, NSW 2006, Australia (e-mail: \{gaoyang.pang, wanchun.liu, daniel.quevedo, branka.vucetic, yonghui.li\}@sydney.edu.au).}
%\thanks{Z. Pang is with the Department of Automation Technology, ABB Corporate Research, 72178 Vasteras, Sweden, and also with the Department of Intelligent Systems, Royal Institute of Technology, 11758 Stockholm, Sweden (e-mail: pang.zhibo@se.abb.com; zhibo@kth.se).}
\thanks{D. Niyato is with the College of Computing and Data Science, Nanyang Technological University, Singapore 639798, (e-mail: dniyato@ntu.edu.sg).}
} 

\maketitle

\begin{abstract}
Wireless Human-Machine Collaboration (WHMC) represents a critical advancement for Industry 5.0, enabling seamless interaction between humans and machines across geographically distributed systems. As the WHMC systems become increasingly important for achieving complex collaborative control tasks, ensuring their stability is essential for practical deployment and long-term operation. 
Stability analysis certifies how the closed-loop system will behave under model randomness, which is essential for systems operating with wireless communications.
However, the fundamental stability analysis of the WHMC systems remains an unexplored challenge due to the intricate interplay between the stochastic nature of wireless communications, dynamic human operations, and the inherent complexities of control system dynamics.
This paper establishes a fundamental WHMC model incorporating dual wireless loops for machine and human control. Our framework accounts for practical factors such as short-packet transmissions, fading channels, and advanced HARQ schemes. We model human control lag as a Markov process, which is crucial for capturing the stochastic nature of human interactions. Building on this model, we propose a stochastic cycle-cost-based approach to derive a stability condition for the WHMC system, expressed in terms of wireless channel statistics, human dynamics, and control parameters. Our findings are validated through extensive numerical simulations and a proof-of-concept experiment, where we developed and tested a novel wireless collaborative cart-pole control system. The results confirm the effectiveness of our approach and provide a robust framework for future research on WHMC systems in more complex environments.

\end{abstract}


\section{Introduction} \label{sec:intro}
\IEEEPARstart{T}{he} Fourth Industrial Revolution, known as Industry 4.0, envisions significantly increased automation and mechanization in manufacturing, driven by rapidly advancing cyber-physical systems (CPS) with minimal human intervention on the factory floor \cite{Industry4.0}. However, many dynamically changing and unforeseen control tasks in manufacturing, such as reconfiguring the production line, are challenging for autonomous machines to handle alone~\cite{RevHRC2}. Therefore, humans are reintroduced to the manufacturing process to collaborate with machines in the fifth industrial revolution, Industry 5.0 \cite{I50}. In the Industry 5.0 era, human-machine collaboration (HMC) emerges as a key enabling technology to boost productivity, efficiency, and sustainability by combining human's creativity, cognitive ability, and dexterity with machine’s strength, precision, and speed \cite{I50}. 
Future wireless communications, e.g., 6G, will be essential to provide high-performance connectivity for humans, machines (including robots), autonomous controllers, and ubiquitous sensors, enabling the flexible, scalable, and low-cost deployment of geographically distributed HMC systems \cite{I50Communication}. Integrating wireless capabilities within an HMC system will unlock the full potential of human-machine collaboration in Industry 5.0, offering unprecedented flexibility and scalability. This wireless HMC (WHMC) framework will serve as the backbone for seamlessly connecting humans, machines, and sensors across geographically distributed environments, enabling real-time collaboration and decision-making.

The main application of WHMC is in collaborative control, where humans and autonomous controllers work together to achieve shared objectives \cite{DefinitionHILC}. WHMC systems enable seamless coordination between human operators and machines, enhancing the efficiency of control tasks. Existing research on WHMC has focused on applied areas such as teleoperation \cite{teleoperation1}, driver assistance systems \cite{ADAS}, and human-machine interaction \cite{HRI}, including scenarios where robots anticipate human intentions and assist in tasks like tool-passing during assembly \cite{DefinitionHILC}. While these efforts have led to successful implementations in specific domains, they often lack the foundational modeling and theoretical analysis needed for broader application \cite{AppliedRes1,AppliedRes2}.

In a WHMC system, stability is a fundamental property that determines whether the controlled states will converge to a steady state and remain bounded under given collaborations. 
Stability analysis is essential for certifying that the closed-loop system will perform safely and effectively, even in the face of human-and-network-induced challenges like random delays and packet loss ~\cite{Liu2023Stability}.
However, the fundamental theories and analytical tools for designing a WHMC system with guaranteed stability are scarce, as this research area is relatively new. Analyzing the stability of a WHMC system presents unique challenges, as it is determined by three tightly coupled domains: wireless communication, human behavior, and control dynamics. Whilst the dynamical properties of individual components are well understood, the stability condition of WHMC systems has yet to be thoroughly investigated.

\subsection{Related Work}
Establishing fundamental theories is important for guiding the systematic design of a desired WHMC system. Researchers have extensively explored theoretical aspects, such as human control modeling, human characteristics modeling, system stability analysis, and wireless networked control.

\subsubsection{Human control modeling}
The primary goal of human control modeling is to mathematically represent how humans perform tasks, enabling machines to understand and adapt to human control policies. This modeling is essential for the design of high-performance machine control systems that can effectively collaborate with human operators. Researchers have proposed various methods to model human control policies. For example, the human operator is often modeled as a classical machine controller, such as linear feedback controller\cite{SimpleHumanModel1,HumanBehaviorMdoeling1}, proportional-integral-feedback controller\cite{SimpleHumanModel2}, and impedance controller\cite{SimpleHumanModel3}. The human control behavior can also be modeled using the crossover-reference model with time-invariant dynamics \cite{HRCinAviation1}, where human operators are characterized as an open-loop transfer function. In addition to the above deterministic models, researchers have also proposed several probabilistic models, such as hidden Markov models (HMMs) \cite{HMMforHuman}, partially observable Markov decision processes (POMDPs) \cite{POMDPforHuman1} and Markov decision processes (MDPs) \cite{MDPforHuman}. Despite significant progress in human control modeling, accurately formulating human control policies mathematically remains a long-lasting unsolved challenge.

\subsubsection{Human characteristics modeling}\label{sec:HMM}
Human characteristics modeling aims to represent the stochastic human traits that implicitly influence the delivery and accuracy of control commands generated by the human decision-making process. These time-varying characteristics include operator workload, proficiency, fatigue, and control lag. For example, the operator's workload can be modeled as a uniform distribution over binary state sets of high and low workload\cite{MDPforHuman}. The operator's fatigue can be modeled as a binary state set (awake or sleepy) with a certain distribution \cite{POMDPforHuman1}. However, the human characteristics in these works are modeled as independent and identically distributed random states. Human characteristics are commonly time-correlated. In order to capture the time-correlated feature, many works model human characteristics as a Markov process \cite{MDPforHuman2} and adopt HMMs to infer human characteristics based on the recorded temporal data \cite{HMMforHuman,HMMforHuman0}. Although using the Markov process to model time-varying human characteristics has garnered significant attention, its application to modeling human control lag has been less considered.\footnote{Modeling the human characteristics impacting the accuracy of human control commands relies on the precise formulation of human control policies, which is a long-lasting unsolved challenge and beyond the scope of our current work \cite{HRCinAviation2}. For a specific collaboration task, the control policy of a human operator commonly remains unchanged in the short term. Thus, we focus on the human control lag, which influences the delivery of human control commands and impacts the collaborative control performance.} The impact of such stochastic human control lag on system performance remains underexplored.

\subsubsection{System stability analysis}
Stability analysis is crucial for designing a WHMC system to operate efficiently and safely. Effective stability analysis requires tractable modeling of the WHMC system. However, most works focus on the fundamental stability analysis of simplified WHMC systems \cite{SimpleHumanModel1,SimpleHumanModel2,SimpleHumanModel3,HRCinDriving1,HRCinAviation1}. In this regard, these works can perform classical analytical frameworks to enable optimal control with a stability guarantee in specific applications, such as irrigation canal \cite{SimpleHumanModel2}, robotic exoskeleton \cite{SimpleHumanModel3}, collaborative driving \cite{HRCinDriving1}, and collaborative piloting \cite{HRCinAviation1}. These limitations make the methodology of most existing works on stability analysis limited to specified control applications, which may weaken the generalization ability of their analytical frameworks. In addition, these systems do not integrate with wireless communication links. Tractable mathematical modeling of advanced WHMC systems with the integration of wireless communication links to establish the stability condition is an unsolved problem.

\subsubsection{Wireless networked control}
Wireless networked control involves integrating autonomous control systems with wireless communication networks. It primarily focuses on establishing systematic theories related to the stability and optimization of state estimation and automatic control over wireless networks \cite{RevWNC,liu2020over,chen2023structure}. Existing research has largely concentrated on developing optimal control algorithms that address the challenges posed by imperfect wireless communication channels, such as errors and delays \cite{WNC1,WNC2}. Some studies investigate the impact of communication protocols and parameters on the stability of automatic control systems \cite{Liu2020WNC1,Liu2020WNC2}. WHMC extends wireless networked control by incorporating human intelligence into the control loop, enhancing system adaptability and performance. While traditional wireless networked control focuses on how communication systems affect control stability, it does not account for the complexities introduced by human operators. Consequently, existing methods in wireless networked control are insufficient for WHMC systems, which require new approaches to address the challenges posed by integrating human factors into the control process. 

\subsection{Motivation}
A WHMC system is significantly more complex than a conventional control system. This complexity arises from the integration of wireless human control loops, the need for collaborative control, and the challenge of addressing time-varying and unforeseen tasks. In a WHMC system, the wireless communication links, the human operator, and the automatic machine controller collectively work to achieve dynamic control objectives under stringent stability constraints. This creates a novel networked topology with tightly coupled wireless human and machine control loops. The system's stability and performance are critically influenced by three factors: wireless communication errors and delays, the stochastic nature of human behavior, and the dynamics of the physical system under control. We name these three factors as the ``three-level dynamics''. Addressing these factors presents a unique challenge in modeling and stability analysis for WHMC systems. To date, the impact of these dynamics on WHMC system stability has not been investigated at all.

Fundamental modeling and analysis of a WHMC system, which features a substantially different control model, requires addressing the following fundamental questions:
\begin{enumerate}
\item How can we achieve tractable mathematical modeling of a WHMC system that effectively captures the three-level dynamics?
\item How can we establish an analytical framework for stability analysis when an accurate mathematical model of the human control policy is unavailable?
\item What are the primary conditions within the three-level dynamics that enable the stable operation of a WHMC system?
\end{enumerate} 

\subsection{Contributions}
In this work, we address the fundamental questions outlined above, and our novel contributions are summarized below.
\begin{enumerate}
    \item \textbf{Novel tractable modeling of the WHMC system}. 
    For the first time, we propose a WHMC model that consists of dual wireless loops, i.e., the machine control loop and the human control loop. In particular, we have taken into account practical wireless communication factors such as short-packet communications, fading channel models and advanced hybrid automatic repeat request (HARQ) schemes for wireless sensors-controller-actuator transmission (referred to as the machine control loop) and sensors-human-actuator transmission (referred to as the human control loop).
    Unlike most existing HMC studies, which typically overlook the temporal variability and stochastic nature of human interactions, we model the dynamics of human control lag as a Markov process.
    %\footnote{We note that some works use simple neural networks \cite{goetschalckx2024computing} and Weibull distributions \cite{mirzaei2013predicting} to model and predict human control lag. These methods assume that the human control lag is independent and identically distributed. However, human control lag is influenced by a variety of time-dependent factors such as learning, fatigue, attention, and physiological states. These factors introduce correlations between response times measured across successive trials as discussed in Section~\ref{sec:HMM}. To capture these time dependencies, models need to account for the fact that past response times, psychological and physiological states can affect future response times. Thus, we model the dynamics of human control lag as a Markov process.} This approach provides a more accurate representation of the variability in human response times, enabling a more realistic and robust analysis of WHMC system performance and stability.
    \item \textbf{Stability analysis of the WHMC system.} Leveraging the proposed system model, we introduce a novel cycle-cost-based approach to derive a sufficient condition for the stochastic stability of the WHMC system for the first time. This stability condition is expressed in terms of wireless channel statistics, human state dynamics, and control system parameters. We thoroughly investigate the structural properties and special cases of the derived stability condition, providing comprehensive analysis and numerical illustrations.
    \item \textbf{Proof-of-concept experiment for the proposed WHMC system.} To demonstrate the advantages of WHMC and validate the developed fundamental theories and analytical tools, a proof-of-concept experiment is conducted. Specifically, we develop and evaluate a wireless collaborative cart-pole control system in terms of control performance and system stability. The experiment confirms the practicality of our approach and provides the validation of the theoretical framework, which is set in 1) and 2).
    %This experimental setup not only confirms the practicality of our approach but also provides a robust framework that can be applied in future research to explore and extend the capabilities of WHMC systems in more complex scenarios.
\end{enumerate}

\textbf{Outline.} The proposed model of the WHMC system is described in Section~\ref{sec:sys}. The stability analysis is presented in Section~\ref{sec:StabilityAnalysis}. A proof-of-concept experiment is demonstrated in Section~\ref{sec:simulation}, followed by conclusions in Section~\ref{sec:conclusion}.

\textbf{Notations.} Matrices and vectors are denoted by capital and lowercase upright bold letters, e.g., $\mathbf{A}$ and $\mathbf{a}$, respectively. $|\mathbf{v}|$ is the Euclidean norm of vector $\mathbf{v}$. $\mathbb{E}\left[\cdot\right]$ is the expectation operator. $\left[\mathbf{A}\right]_{i,j}$ denotes the element at $i$-th row and $j$-th column of a matrix $\mathbf{A}$. $(\cdot)^\top$ is the vector or matrix transpose operator. $\mathbb{R}$ and $\mathbb{N}$ denote the sets of real numbers and positive integers, respectively. $\mathbb{N}_{0}$ denotes the non-negative integers. 

\section{WHMC System} \label{sec:sys}
\subsection{Control System Dynamics}
\begin{figure}[!t]
	\centering\includegraphics[width=2.9in]{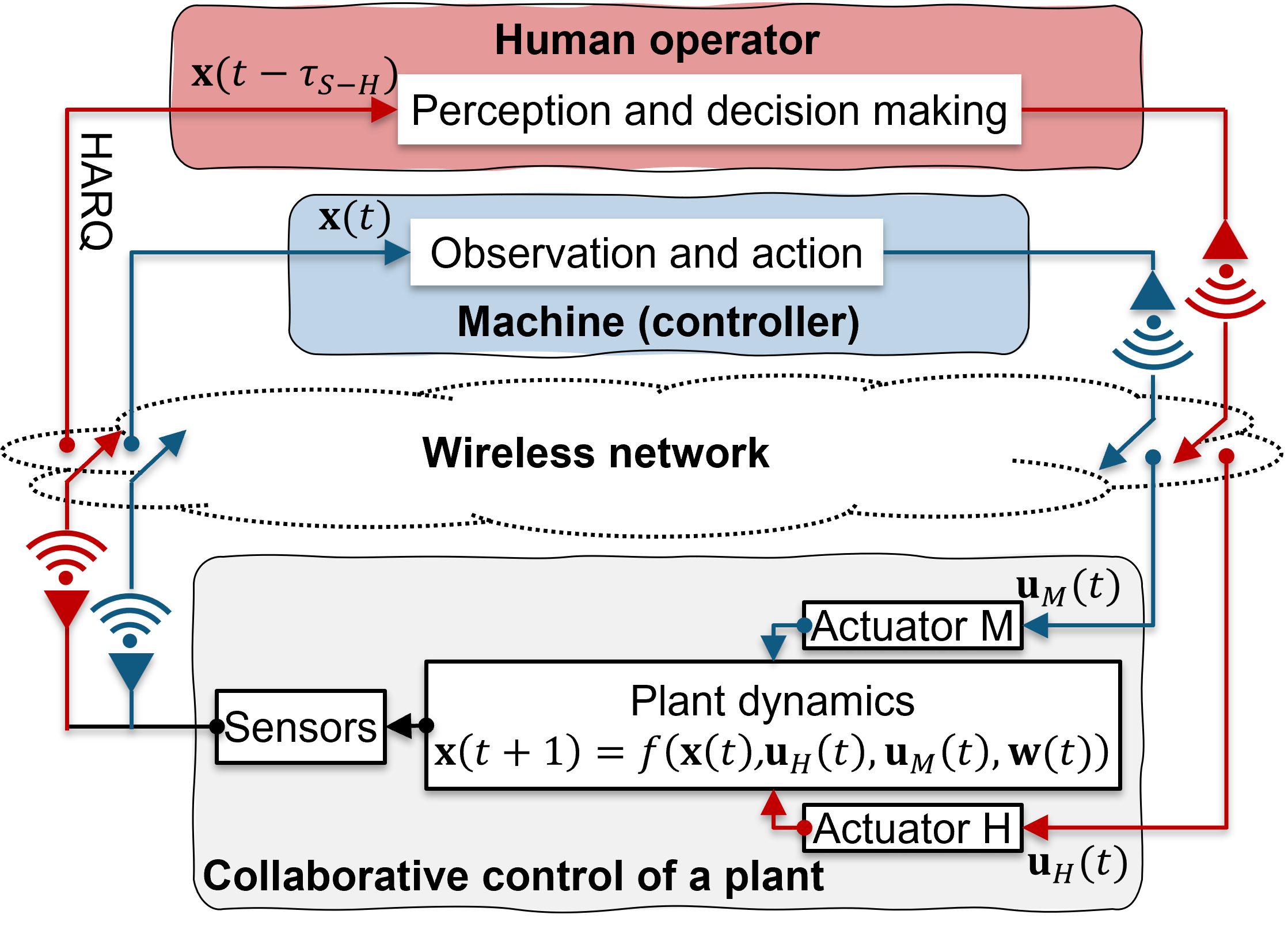}
	\vspace{-0.4cm}
	\caption{Illustration of the WHMC system, consisting of two types of control loops, i.e., the machine control loop and the human control loop.}
	\label{fig:SysModel}
	\vspace{-0.5cm}
\end{figure}
We consider a WHMC system consisting of a dynamic plant, two actuators, an autonomous controller (i.e., a machine), and a human operator, as shown in Fig.~\ref{fig:SysModel}. The sensors attached to the plant send state measurements to the remote controller and the human operator. These two agents then send their individual control signals to the corresponding actuators in order to complete a collaborative control task of the plant. All information for sensing and control is exchanged via four wireless links: sensor-human (SH) uplink, sensor-controller (SC) uplink, human-actuator (HA) downlink, and controller-actuator (CA) downlink. Such a system model has two types of control loops, i.e., the machine control loop and the human control loop. It is abstracted from the existing visions of HMC systems, e.g., homecare robotic systems for Healthcare 4.0 \cite{HRCinHomecare}, factory edge robotic systems for Industrial 5.0 \cite{I50Communication}, collaborative surgery in healthcare \cite{CoDesign}, collaborative piloting in aviation \cite{HRCinAviation1}, and collaborative driving in a vehicle \cite{HRCinDriving1}. These systems require a human operator to control an actuator as well as collaborate with other machine-controlled actuators. 

Having two loops in parallel allows one to clearly distinguish between human and machine contributions and enables individual analysis of each loop's dynamics and their interactions. Our model can also adjust the degree of influence each loop has, allowing for a spectrum of control schemes, such as human-in-the-loop, supervisory, and shared control.\footnote{For example, if the time period of a human control loop is far longer than that of a machine control loop, our model becomes supervisory control, where the machine is predominantly responsive. If the time period of a machine control loop is longer than that of a human control loop, it can be seen as human-in-the-loop control, where the human operator is predominantly responsive. If the time period of a machine control loop is close to that of a human control loop, our model encompasses shared control, where both the human operator and the machine contribute significantly.} 

The plant dynamics is modeled as a nonlinear discrete time-invariant system
\begin{equation}\label{NLTI}
    \mathbf{x}(t+1) =f(\mathbf{x}(t), \mathbf{u}_{H}(t), \mathbf{u}_{M}(t),\mathbf{w}(t)),
\end{equation}
where $t$ is the time index given the sampling period $T_s$; $\mathbf{x}(t)\in\mathbb{R}^{l_s}$ is the plant state vector at time $t$; $\mathbf{u}_{H}(t) \in \mathbb{R}^{l_h}$ and $\mathbf{u}_{M}(t) \in \mathbb{R}^{l_m}$ are the corresponding human control input and machine control input, respectively; $\mathbf{w}(t)\in \mathbb{R}^{l_w}$ is the plant disturbance. The control algorithms for generating control inputs will be presented later in this section.

\subsection{Wireless Control Loops}\label{sec:channel}
The temporal operation of the two control loops is shown in Fig.~\ref{fig:TwoTimeScaleControlLoops}. We assume block Rayleigh fading channels, where the channel characteristics remain constant during each time slot but change independently from one time slot to the next. 
\begin{figure}[!t]
	\centering\includegraphics[width=3.5in]{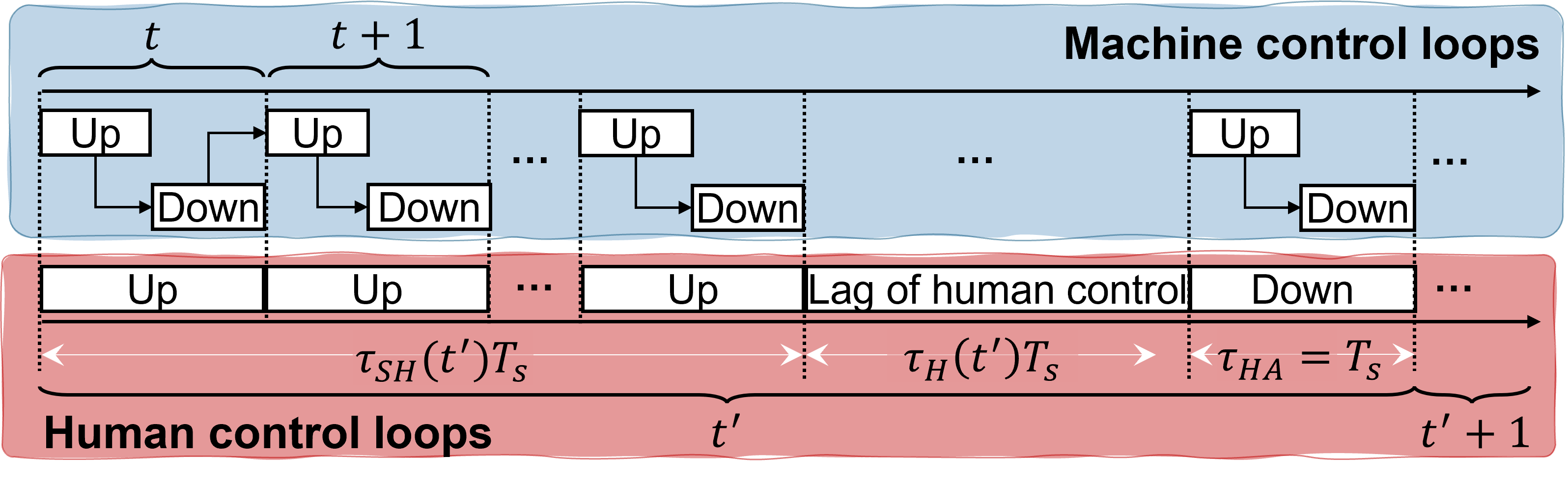}
	\vspace{-0.8cm}
	\caption{Temporal operation of the two control loops.}
	\label{fig:TwoTimeScaleControlLoops}
	\vspace{-0.5cm}
\end{figure}

\subsubsection{Machine control loops}
Each machine control loop takes a single time step, i.e., the period of a machine control loop is $T_s$, and consists of a pair of SC uplink and CA downlink transmissions. If the SC packet is not detected successfully, there is no CA transmission scheduled as the controller has no instantaneous plant state information. We consider short-packet transmissions for low-latency communications~\cite{Polyanskiy2010BLER}. The computation time for generating a control signal is usually much shorter than the transmission delay and thus is omitted~\cite{OmittedCompute1,OmittedCompute2}.
\emph{A machine control loop is closed only when both the SC and CA transmissions within it are successful.}

\subsubsection{Human control loops} \label{sec:introHumanCtrlLoop}
Each human control loop period is delineated by an HA downlink transmission, as illustrated in Fig.~\ref{fig:TwoTimeScaleControlLoops}. A human control loop contains multiple SH uplink transmissions, a human control procedure, and an HA downlink transmission. A downlink transmission attempt marks the end of one human control loop period and the beginning of the next.
Each period starts from a new packet transmission from the sensors, which contains the current plant state measurement. If the transmission fails, then a retransmission takes place using a HARQ protocol. In instances where a given maximum number of retransmissions $N$ has been reached, a new transmission is triggered.\footnote{Unlike machine control loops, the lag of human control, which captures the delay in human decision-making, is significantly longer than the transmission delay~\cite{LagOverTransmission}. This results in frequent machine control actions and infrequent human interventions. In contexts where the lag of human control is substantial, the transmission delay becomes relatively insignificant. Consequently, retransmissions are used to improve transmission reliability, as a longer transmission delay caused by retransmissions does not notably affect the overall human control process.} The human operator generates a control command after receiving a successful packet, and then sends the command to the actuator.
There is no retransmission for the HA and CA downlinks, since retransmissions lead to unpredictable delays, making the generated time-sensitive control command useless.
Let $t'$ denote the human control loop index. Then, the transmission delay for the SH and HA transmissions are $\tau_{SH}(t')$ and $\tau_{HA}=1$, respectively, and the lag of human control is $\tau_{H}(t') \in \mathcal{S} \triangleq \{ 1,2,\dots,\tau_{\mathrm{max}}\}$.
In particular, $\{\tau_{H}(t')\}$ is modeled as a finite state Markov chain with a transition probability matrix $\mathbf{M}$, where $p_{i,j} \triangleq [\mathbf{M}]_{i,j}$. A shorter lag of human control leads to better control performance. The stationary distribution of $\tau_{H}(t')$ is given as
\begin{equation} \label{eq:DistriHumanState}
v_{k} \triangleq \mathbb{P}\!\left[ \tau_{H}(t')\!=\!k \right]\!, 1 \leq k \leq \tau_{\text{max}}.
\end{equation}
We assume each transmission in a human control loop takes one time step because human-type communication generally requires a larger packet length than machine-type communication \cite{LargePacket}.
Considering the random period of each human control loop, we define $\kappa(t')$ as the starting time slot of the $t'$-th human control loop.
\emph{The human control loop is closed once the HA transmission is successful.}

\subsection{Control Algorithms}\label{sec:CollaborativeControlModel}
Due to packet detection errors, the sensor's packet for the remote controller may not be received by the remote controller, and the machine control input may not be received by the actuator at every time step.
Let the binary variables $\zeta_{SC}(t), \zeta_{CA}(t),\zeta_{SH}(t'),\zeta_{HA}(t') \in \{1,0\}$ denote the transmission success and failure of the corresponding channel at $t$, respectively.
The \textbf{machine control input} at $t$ is given as
\begin{equation}\label{eq:MachineControl}
\mathbf{u}_{M}(t)=\begin{cases}
f_M \left(\mathbf{x}(t) \right), & \text { if } \zeta_{CA}(t)\zeta_{SC}(t)=1,\\
\mathbf{0}, & \text {otherwise,}
\end{cases}
\end{equation}
where $f_M(\cdot)$ is the machine control policy.
Hence, only a pair of successful uplink and downlink transmissions can generate an effective control input, closing the machine control loop.
%Therefore, only a pair of successful uplink and downlink transmissions can generate an effective control input and ``closes" the machine control loop.

From the definition of human control loops, a human control input can only be available at the beginning of each control loop. Considering the random delay of SH transmissions and human decision-making, the \textbf{human control input}  at $t$ is
\begin{equation}\label{eq:HumanControl}
\mathbf{u}_{H}(t)\!=\!\begin{cases}
    f_{H} (\mathbf{x}(t\!-\!\tau_{SHA}(t')),\!& \text {for } t\!=\!\kappa(t') \text { and } \zeta_{HA}(t')\!=\!1 \\
    \mathbf{0},\!& \text {otherwise,}
\end{cases}
\end{equation}
where $f_H(\cdot)$ is the human control policy and 
\begin{equation} \label{eq:DuraH}
\tau_{SHA}(t') \triangleq \underbrace{\tau_{SH}(t')}_\text{SH Tx. delay}+ \underbrace{\tau_{H}(t')}_\text{Human control lag}+\underbrace{\tau_{HA}}_\text{HA Rx. delay}.
\end{equation}
As an accurate model of the human control policy is unavailable, we propose an analytical framework for stability analysis without using specific control policies of human and machine, but using their control significance in the next section.
%Several machine control cycles, i.e., $\kappa(t'-1) < t < \kappa(t')$, may have been completed before the control input of the human operator updates at the actuator due to the packet loss. Thus, the control inputs of the human operator are zeros between its two consecutive updates.

\section{Stability Analysis} \label{sec:StabilityAnalysis}
From \eqref{eq:MachineControl}--\eqref{eq:DuraH}, we see that the two control loops can be either open or closed due to the packet loss and delays, which may cause instability of the WHMC system. In this section, we derive the stability condition of the proposed WHMC system by taking into account the randomness in wireless communications and human decision-making.
Since only closed control loops generate effective control inputs that regulate plant state and affect stability, we analyze the statistics of the stochastic closed (and open) control loop first.

\subsection{Stochastic Control Loop Analysis} \label{sec:LoopAnalysis}
%In the proposed WHMC system, there are three challenging issues for stability analysis, including the different time scales of the machine control loop and the human control loop, the open loop probabilities of machine control and human control, and the stochastic time duration of the human control loop. Before deriving the stability conditions of the WHMC system, a deep analysis is required to derive the open loop probabilities of human and machine, the distribution of human control loop lengths, and the length distribution of consecutive open human control loops.
\subsubsection{Open loop probabilities of human and machine control}
Let $\gamma_{HA}(t)$, $\gamma_{SH}(t)$, $\gamma_{CA}(t)$, and $\gamma_{SC}(t)$ denote the signal-to-noise ratio (SNR) of received packets in HA, SH, CA, and SC channels, respectively. %Thus, the receiving SNR follows an exponential distribution. 
Given the packet length $l_p$ (i.e., the number of symbols per packet), the number of data bits $b$ in the packet, and the SNR $\gamma$ of the packet, we have the approximated decoding error probability of a packet as \cite{PangFBL}
\begin{equation}\label{BLER:NoHARQ}
\varepsilon\left(\gamma\right) \approx \mathcal{Q}\left(\frac{\mathcal{C}\left(\gamma\right)-\frac{b}{l_p}}{\sqrt{\frac{\mathcal{V}\left(\gamma\right)}{l_p}}}\right),
\end{equation}
where $\mathcal{C}(\gamma)=\log_2{(1+\gamma)}$ and $\mathcal{V}(\gamma)=(1-(1+\gamma)^{-2})(\log_2{e})^2$ are the Shannon capacity and the channel dispersion, respectively, and $\mathcal{Q}(x)=(\frac{1}{\sqrt{2\pi}})\int_{x}^{\infty}{e^{-\frac{t^2}{2}}\text{d}t}$ is the Gaussian Q-function. 

The probability of the machine control operating in an open loop at time $t$ can be obtained as
\begin{equation}\label{Pr_OpenMachineLoop}
\begin{aligned}
    p_M(t) &= \mathbb{P}\!\left[ \zeta_{SC}(t) = \zeta_{CA}(t) = 1 \right] \\ 
    &= 1 - \left(1 - \varepsilon\left(\gamma_{SC}\left(t\right)\right)\right) \left(1 - \varepsilon\left(\gamma_{CA}\left(t\right)\right)\right).
\end{aligned}
\end{equation}
The expectation of \eqref{Pr_OpenMachineLoop} with respect to $\gamma_{SC}\left(t\right)$ and $\gamma_{CA}\left(t\right)$ is denoted as $\Bar{p}_M$, and can be obtained by
\begin{equation}\label{ExPr_OpenMachineLoop}
 \Bar{p}_M\!\triangleq\!\mathbb{E}\!\left[p_M(t)\right]\!=\!1\!-\!\left(1\!-\!\mathbb{E}\!\left[\varepsilon\!\left(\!\gamma_{SC}(t)\right)\right]\right)\!\left(\!1\!-\!\mathbb{E}\!\left[\varepsilon\!\left(\!\gamma_{CA}(t)\right)\right]\right)\!.  
\end{equation}

Since each human control loop contains a successful SH packet, the probability of an open human control loop only depends on the HA transmission and is given by 
\begin{equation}\label{Pr_OpenHumanLoop}
    p_H(t') = \mathbb{P}\!\left[ \zeta_{HA}(t')\!=\!0 \right] = \varepsilon\left(\gamma_{HA}(\kappa(t'+1)-1)\right).
\end{equation}
The expectation of \eqref{Pr_OpenHumanLoop} with respect to $\gamma_{HA}(\kappa(t')-1)$ is denoted as $\Bar{p}_H$, and can be obtained by
\begin{equation}\label{ExPr_OpenHumanLoop}
    \Bar{p}_H \triangleq \mathbb{E}\!\left[p_H(t')\right] = \mathbb{E}\!\left[\varepsilon\left(\gamma_{HA}(\kappa(t')-1)\right)\right].
\end{equation}

\begin{figure}[!t]
	\centering\includegraphics[width=3.3in]{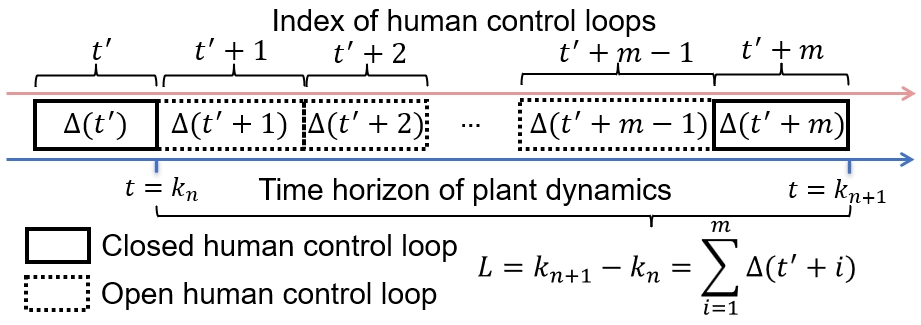}
	\vspace{-0.4cm}
	\caption{Illustration of the time horizon of plant dynamics between two adjacent closed human control loops. $\Delta(\cdot) = \tau_{SH}(\cdot) + \tau_{H}(\cdot) + \tau_{HA}$}
	\label{fig:CycleBasedAnalysis.png}
	\vspace{-0.5cm}
\end{figure}

\subsubsection{Distribution of the SH delay}
The duration of a human control loop $\tau_{SHA}(t')$ in \eqref{eq:DuraH} includes the SH channel delay $\tau_{SH}(t')$, human control lag $\tau_{H}(t')$, and the HA channel delay $\tau_{HA}$. The HA channel delay is constant across human control loops, while the human control lag is time-correlated across human control loops due to the Markovian property. The SH channel delay is attributed to the HARQ and i.i.d across all human control loops. We analyze the distribution of the SH channel delay before proceeding with the distribution of the duration of consecutive time steps where the human control loop is open. 
We consider the following three types of HARQ schemes for the SH channel, including Type I HARQ (TI-HARQ), Chase Combing HARQ (CC-HARQ), and Incremental Redundancy HARQ (IR-HARQ).\footnote{In TI-HARQ, the packet is re/transmitted for all re/transmissions, and all erroneously decoded packets are discarded at the receiver side. All decoding attempts during re/transmissions of the packet are independent. In CC-HARQ, all erroneously decoded packets in previous re/transmissions are saved and their signals are combined together as a single strengthened signal for decoding. In IR-HARQ, the packet in each re/transmission is a punctured version of a low-rate mother packet. If errors occur, it only retransmits the additional redundancy for the previous uncorrectable packets. The newly received redundancy is combined with the previously received packets to construct a packet with a longer length for decoding.}

The number of re/transmission attempts is $r \in \{1,2,\dots,N\}$. Let $\boldsymbol{\gamma}_{r}(\kappa(t'-1),r)$ denote the set of experienced SNRs during $r$ re/transmission attempts, that is 
\begin{equation} \label{eq:SetOfGamma}
\boldsymbol{\gamma}_{r}(\kappa(t'\!-\!1),r)\!\triangleq\!\{\!\gamma_{SH}(\kappa(t'\!-\!1)\!),\dots,\gamma_{SH}(\kappa(t'\!-\!1)\!+\!r\!-\!1\!)\!\}.
\end{equation}
The decoding error probability of the packet after $r$ re/transmission attempts $\Theta(r)$ is an expectation over \eqref{eq:SetOfGamma}, and can be approximated as \cite{HARQ,HARQ1,HARQ2}
\begin{equation}\label{EqBLERinSHchannel}
\begin{aligned}
    \Theta(r)\!&\triangleq\!\mathbb{P}\!\left[ \zeta_{SH}(t')\!=\!0 \mid\boldsymbol{\gamma}_{r}(\kappa(t'-1),r)\right] \\
    &\approx\!\begin{cases}\!
    \prod_{i=0}^{r-1}\varepsilon\left(\gamma_{SH}(\kappa(t'-1)+i)\right)\!, &\!\text{\small TI-HARQ}, \\
    \varepsilon\left(\sum_{i=0}^{r-1}\gamma_{SH}(\kappa(t'-1)\!+\!i)\!\right)\!, &\!\text{\small CC-HARQ}, \\
    \mathcal{Q}\!\!\left(\!\!\frac{\sum_{i=0}^{r-1}\mathcal{C}\left(\gamma_{SH}(\kappa(t'\!-\!1)+i)\right)-\frac{b}{l_p}}{\sqrt{\frac{\sum_{i=0}^{r-1}\mathcal{V}\left(\gamma_{SH}(\kappa(t'-1)+i)\!\right)}{l_p}}}\!\!\right)\!\!,&\!\text{\small IR-HARQ}.
\end{cases}
\end{aligned}
\end{equation}
To facilitate our subsequent analysis, we assume that all packets have the same length $l_p$. For CC-HARQ, since the channel gain is exponentially distributed, $\sum_{i=0}^{r-1}\gamma_{SH}(\kappa(t'+i))$ is gamma distributed with the probability distribution function \cite{GammaDistribution}
\begin{equation}\label{PDFforSumRayleigh}
\mathbb{P}\!\left[ \sum_{i=0}^{r-1}\gamma_{SH}(\kappa(t'+i)) = \hat{\gamma} \right]\! = \frac{\frac{1}{\Bar{\gamma}^{r}} \hat{\gamma}^{r-1} e^{-\frac{\hat{\gamma}}{\Bar{\gamma}}}}{\left( r-1 \right)!},
\end{equation}
where $\Bar{\gamma}$ is the mean of the exponential distribution.
Thus, for TI- and CC-HARQ, $\Theta(r)$ is obtained by leveraging \eqref{BLER:NoHARQ}, \eqref{EqBLERinSHchannel}, and \eqref{PDFforSumRayleigh}. For IR-HARQ, $\Theta(r)$ can be determined by Monte Carlo simulations. 
\begin{figure*}[!t]
\normalsize
%\vspace{-0.5cm}
%\hrulefill
\begin{equation}\label{eq:S-Hdelay}
\!w_{k} \triangleq \mathbb{P}\!\left[ \tau_{SH}(t')\!=\!k \right]\!
\!=\!\begin{cases}
    \!\left(\Theta(N)\right)^{q} \!\left(1\!-\!\Theta(k-qN)\right), &\!\text{for } k\!=\!qN\!+\!1, \\
    \!\left(\Theta(N)\right)^{q} \!\left(\Theta(k-qN-1)\!-\!\Theta(k-qN) \right), &\!\text{for } qN\!+\!2\!\leq k\!\leq (q\!+\!1)N, \\
    0, & \text{otherwise.}
   \end{cases} 
\end{equation}
%\hrulefill
\vspace{-0.8cm}
\end{figure*}
The delay induced by the SH transmission period $\tau_{SH}(t')$ is $k \in \mathbb{N}_0$. Note that the SH transmission period may contain multiple $N$-trails of the retransmission process as described in Section~\ref{sec:introHumanCtrlLoop}, and the number of experienced $N$-trails is $q \in \mathbb{N}_0$. 
The probability distribution of $\tau_{SH}(t')$ is then given as \eqref{eq:S-Hdelay}.
When $q = 0$, the SH transmission is successful in the first $N$-trials. In this case, if $2 \leq k \leq N -1$, $\tau_{SH}(t')\!=\!k$ indicates that the first $k - 1$ trials have failed and the $k$-th transmission attempt is successful. When $q > 0$, the SH transmission is successful in the $\left(q + 1\right)$th $N$-trials, while the former $q$th $N$-trials are decoded erroneously.

\subsubsection{Time interval distribution between consecutive closed human control loops}
We denote the starting time of the $n$th closed human control loop as $t = k_n$, as shown in Fig.~\ref{fig:CycleBasedAnalysis.png}. 
Let $L$ and $M$ denote time steps and the numbers of (open or closed) human control loops between $k_n$ and $k_{n+1}$, respectively, i.e., 
\begin{equation} \label{eq:TimeInterval}
    L \triangleq \sum_{i=1}^{M} \tau_{SH}(t' + i) + \sum_{i=1}^{M} \tau_{H}(t' + i) + M,
\end{equation}
where $t'$ is the index number of the $n$th closed human control loop among all the loops.
The probability distribution of $L$ in \eqref{eq:TimeInterval} can be expressed as 
\begin{equation}\label{eq:MarkovCompelexity}
    z_{l} \triangleq \mathbb{P}\!\left[L \!=\!l \right]\!= \sum_{m = 1}^\infty  \mathbb{P}\!\left[L = l \mid M = m \right] \mathbb{P}\!\left[ M\!=\!m \right].
\end{equation}
The probability distribution of the number of consecutive open human control loops in \eqref{eq:MarkovCompelexity} can be expressed as 
\begin{equation}\label{eq:ProbM}
    \mathbb{P}\!\left[ M\!=\!m \right]\!= (1-\Bar{p}_H) (\Bar{p}_H)^{m-1}, m \in \mathbb{N},
\end{equation}
where $\Bar{p}_H$ is defined in \eqref{ExPr_OpenHumanLoop}.
The time interval distribution of $L$ under the condition with $m$ open human control loops in \eqref{eq:MarkovCompelexity} consists of two independent and stochastic parts, i.e., the total delay induced by SH channel and human control lag. In the following, we analyze the conditional probabilities of the two parts.
The conditional probabilities of the delay induced by the SH channel can be expressed as
\begin{equation} \label{eq:DistriS-H}
\begin{aligned}
    w_{k,m} & \triangleq \mathbb{P}\!\left[\sum_{i=1}^{m} \tau_{SH}(t' + i) = k \mid M = m \right] \\
    &= \begin{cases}
        \sum_{i=1}^{k} w_{i,m-1}w_{k-i+1}, & \text{ for } m > 1, \\
        w_k, & \text{ for }  m = 1,
        \end{cases}
\end{aligned}
\end{equation}
where $w_k$ is defined in \eqref{eq:S-Hdelay}.
The conditional probabilities of the delay induced by the human control lag can be expressed as \eqref{eq:DistriH-A},
\begin{figure*}[!t]
\normalsize
%\vspace{-0.5cm}
%\hrulefill
\begin{equation}\label{eq:DistriH-A}
\begin{aligned}
    v_{k,m} \triangleq \mathbb{P}\!\left[\tau_{H}(t' + i) = k \mid M = m \right]\! &= \begin{cases}
               \sum_{\delta_1 + \dots + \delta_m = k} \mathbb{P}\!\left[\tau_{H}(t'+1) = \delta_1, \dots, \tau_{H}(t' + m) = \delta_m\right], & \text{ for }  m > 1,\\
               v_k,  & \text{ for }  m = 1.
           \end{cases} \\
           & = \begin{cases}
               \sum_{\delta_1 + \dots + \delta_m = k} v_{\delta_1}p_{\delta_1,\delta_2}p_{\delta_2,\delta_3} \dots p_{\delta_{m-1},\delta_m}, & \text{ for }  m > 1,\\
               v_k,  & \text{ for }  m = 1.
           \end{cases} 
\end{aligned}
\end{equation}
\hrulefill
\vspace{-0.6cm}
\end{figure*}
where $\delta_m \in \mathcal{S}$, $v_k$ is defined in \eqref{eq:DistriHumanState}, and $p_{\delta_{m-1},\delta_m} = [\mathbf{M}]_{\delta_{m-1},\delta_m}$.
Then, the time interval distribution under the condition with $m$ open human control loops is
\begin{equation} \label{eq:CondiProb}
\begin{aligned}
    z_{l,m} & \triangleq \mathbb{P}\!\left[L = l \mid M = m \right] \\
    &=\begin{cases}
    \sum_{k=1}^{l-m} w_{k,m}v_{l -m -k +1,m}, & \text{for } l>m \\
    0, & \text{otherwise,}
    \end{cases}
\end{aligned}
\end{equation}
where $w_{k,m}$ and $v_{k,m}$ are conditional probabilities defined in \eqref{eq:DistriS-H} and \eqref{eq:DistriH-A}, respectively.
In summary, by using \eqref{eq:DistriS-H} and \eqref{eq:DistriH-A}, we can obtain \eqref{eq:CondiProb}. By substituting \eqref{eq:ProbM} and \eqref{eq:CondiProb} into \eqref{eq:MarkovCompelexity}, we can obtain the time interval distribution between consecutive closed human control loops.

\subsection{Stability Condition of WHMC}
Lyapunov functions are powerful tools used for stability analysis in dynamic systems without needing explicit control policies. A function $V: \mathbb{R}^{l_s}\rightarrow\mathbb{R}_{\geq0}$ is said to be a Lyapunov-like function, if $V(0) = 0$, $V\left(\mathbf{x}(t)\right)>0$ for $\mathbf{x}(t)\neq0$, and $\lim_{||\mathbf{x}(t)||\rightarrow\infty}V\left(\mathbf{x}(t) \right)= \infty$. It is a scalar function that can be treated as a cost function associated with the system state $\mathbf{x}(t)$. For example, the function $V(\mathbf{x}(t))$ can be the magnitude of the input vector $\mathbf{x}(t)$. The dynamic system is stable if the expected cumulative cost over an infinite time horizon remains bounded. Thus, we have the following definition.
\begin{definition}[Stochastic Stability \cite{LiuStability,Assumption1,Assumption2}]\label{def:stability}
	\normalfont
	The wireless networked human-machine collaborative system is stochastically stable, if for some Lyapunov-like functions $V$: $\mathbb{R}^{l_s}\rightarrow\mathbb{R}_{\geq0}$, the expected value $\sum_{t=0}^{\infty} \mathbb{E}\left[V\left(\mathbf{x}(t)\right)\right] < \infty$. 
\end{definition}
 
From \eqref{Pr_OpenMachineLoop}, \eqref{Pr_OpenHumanLoop} and \eqref{eq:DistriS-H}, we note that the WHMC system randomly switches between the following four cases: 1) Case one: both the machine control loop and the human control loop are closed; 2) Case two: only the machine control loop is closed; 3) Case three: only the human control loop is closed; and 4) Case four: both the machine control loop and the human control loop are open. We next examine the stability condition taking into account each individual case.

For tractable analysis, we make the following assumption.
\begin{assumption}[Lyapunov-Like Function Gains]\label{assump:controlGains}
    \normalfont
    There exists a Lyapunov-like function $V$: $\mathbb{R}^{l_s}\rightarrow\mathbb{R}_{\geq0}$, non-negative control system parameters $\alpha_{HM}\in \mathbb{R}_{\geq0}$, $\alpha_{M}\in \mathbb{R}_{\geq0}$, $\alpha_{H}\in \mathbb{R}_{\geq0}$, and $\alpha \in \mathbb{R}_{>0}$, such that for all $\mathbf{x}(t)$ following \eqref{NLTI} and the initial plant state satisfying $\mathbb{E}\left[V (\mathbf{x}(0))\right] < \infty$, we have
    \begin{equation} \label{eq:HMCsystemParameters}
            V(\mathbf{x}(t+1)) \leq \begin{cases}
                \alpha_{HM} V(\mathbf{x}(t)), & \text{for case one}, \\
                \alpha_{M} V(\mathbf{x}(t)), & \text{for case two}, \\
                \alpha_{H} V(\mathbf{x}(t)), & \text{for case three}, \\
                \alpha V(\mathbf{x}(t)), & \text{for case four}.
            \end{cases}
    \end{equation}
\end{assumption}
Assumption~\ref{assump:controlGains} bounds the one-step cost function ratio between $V(\mathbf{x}(t+1))$ and $V(\mathbf{x}(t))$ in the four cases based on the Lyapunov gains, $\alpha_{HM}$, $\alpha_{M}$, $\alpha_{H}$, and $\alpha$. Note that Lyapunov gains are often assumed in non-linear system stability 
analysis~\cite{LiuStability,Assumption1,Assumption2}.
If a ratio is less than $1$, then the cost decreases; otherwise, it increases.
Considering extreme cases, if all ratios in the four cases are less than 1, the WHMC system is directly stabilized, as the cost in all cases decreases over time. Conversely, if all ratios are significantly larger, the system may not stabilize according to Definition~\ref{def:stability}.
The control system parameters $\alpha_{HM}$, $\alpha_{M}$, $\alpha_{H}$, and $\alpha$ are determined by the plant dynamics~\eqref{NLTI} and the control algorithms \eqref{eq:MachineControl} and \eqref{eq:HumanControl}.

%In real applications, these control system parameters can be estimated by pilot interactions (see Section~\ref{EstimationForStabilityAnalysis}). 

\subsubsection{Stability condition}\label{sec:StochasticAnalysis}
In the following, we propose a stochastic cycle-cost-based approach to obtain sufficient stability conditions for the WHMC system.
\begin{theorem}\label{threm:stability}
	\normalfont
    The plant of the WHMC system defined in Section~\ref{sec:sys} is stochastically stable if 
    \begin{equation} \label{eq:Stability}
   \mathbb{E}\left[\!\left(\alpha_M\!\left(\!1\!-\!\Bar{p}_M\right)\!+\!\alpha \Bar{p}_M\!\right)^{L}\!\right]\!\left(\!\alpha_{HM}\!\left(\!1\!-\!\Bar{p}_M\!\right)\!+\!\alpha_H\Bar{p}_M\!\right)\!<\!1,
    \end{equation}
    where $\Bar{p}_M$ is the expected probability of an open machine control loop defined in \eqref{ExPr_OpenMachineLoop}; the control system parameters $\alpha_{HM}$, $\alpha_{M}$, $\alpha_{H}$, and $\alpha$ are defined in Assumption~\ref{assump:controlGains}; $L$ is the random time interval between consecutive closed human control loops with the probability distribution defined in \eqref{eq:MarkovCompelexity}.
\end{theorem}
\begin{proof}(Main ideas)
	We investigate the stability condition of the WHMC system defined in \eqref{NLTI} by following the stability analysis framework adopting Lyapunov-like functions \cite{LiuStability,Assumption1,Assumption2}.\footnote{The methods in \cite{LiuStability,Assumption1,Assumption2} are not directly applicable, as the control process involves human control operations with a Markovian lag model.}
    Since human control is less frequent than machine control, it is convenient to focus on the plant events in which the actuator received human control commands.
    Therefore, the control process is divided by the closed-human-control-loop events. We name the time interval between consecutive closed human control loops as a cycle within the control process, and the sum of stochastic costs in a cycle is a cycle cost. Thus, the total cost of the control process is the sum of all cycle costs. The stability is equivalent to the bounded sum of all cycle costs, according to Definition~\ref{def:stability}.
    To prove the stability condition, we first analyze a stochastic cycle cost, where only case two and case four defined in Assumption~\ref{assump:controlGains} exist. It depends on the number of these two cases conditioned on the open loop probabilities and the time interval distribution presented in Section~\ref{sec:LoopAnalysis}. Then, we analyze the sum of stochastic cycle costs to the infinity cycles by further considering case one and case three defined in Assumption~\ref{assump:controlGains}. Finally, we derive the stability condition by making the sum of the stochastic cycle costs bounded as Definition~\ref{def:stability}. See Appendix~\ref{proofThrem1} for detailed proof.
\end{proof}

Sufficient conditions in stability analysis are critical because they provide guarantees that the system will be stable under the specific assumption. They are thus preferred since they give engineers and researchers a clear set of criteria to design and analyze their systems safely. 
The stability condition of the WHMC systems depends on the wireless communication parameters, i.e., the open loop probabilities of human and machine control $\Bar{p}_H$ and $\Bar{p}_M$, 
the control system parameters, i.e., $\alpha_{HM}$, $\alpha_{M}$, $\alpha_{H}$ and $\alpha$, and the Markov human state transition rule $\mathbf{M}$. In particular, $\Bar{p}_H$ and $\mathbf{M}$ impact the distribution of $L$, which further affect the stability condition. 
The condition indicates that if the WHMC systems exhibit high dynamics (i.e., the plant state changes significantly even with very small control input), the human operator experiences fatigue with a high control lag, and the open-loop probability is high, then the WHMC system becomes difficult to stabilize through collaboration. 
%We note that instability of control systems often happens due to delay. In addition to the delay induced by the HARQ, there are many other causes for delays, such as queueing and congestion, which will be considered in our future work.

\subsubsection{Stability region}
The stability region in WHMC systems defines the range of system parameters that ensure stable operation, as per Theorem~\ref{threm:stability}.
The boundary of this stability region represents the critical limits beyond which the system may become unstable. The properties of this boundary are elucidated next.
\begin{corollary}\label{corollary:linearity}
\normalfont
Given the WHMC stability condition in Theorem~\ref{threm:stability},

(i) the stability region boundary in terms of $\alpha_{HM}$ and $\alpha_H$ is linear:
\begin{equation}\label{eq:linearEqual1}
    \alpha_{HM}  = - \frac{\Bar{p}_M }{\left( 1-\Bar{p}_M\right)}\alpha_H + \frac{1}{\mathbb{E}\left[ \Omega^{L} \right] \left( 1-\Bar{p}_M\right)},
\end{equation}
where $\Omega \triangleq \alpha_M \left( 1-\Bar{p}_M\right) + \alpha \Bar{p}_M$; 

(ii) the stability region boundary in terms of $\alpha_M$ and $\alpha$ is linear:
\begin{equation} \label{eq:linearEqual2}
    \alpha_M  = - \frac{\Bar{p}_M}{1-\Bar{p}_M}\alpha   + \frac{1}{\left( 1-\Bar{p}_M\right) \bar{L}}\sum_{l = 1}^{\bar{L}}\left( \bar{L} \Lambda \mathbb{P}\left[L = l\right]\right)^{-l},
\end{equation}
where $\bar{L}\triangleq \mathbb{E}\left[L\right]$ and $\Lambda \triangleq \alpha_{HM} \left( 1-\Bar{p}_M\right) + \alpha_H \Bar{p}_M$; 

(iii) the stability region boundaries, in terms of the other four possible pairs of control system parameters, i.e., $\alpha_{HM}$, $\alpha_{M}$, $\alpha_{H}$, and $\alpha$, are concave.
\end{corollary}
\begin{proof}
    See Appendix~\ref{proofCoro1}.
\end{proof}

As illustrated in Fig.~\ref{fig:StabilityRegion}, a linear stability region (e.g., Corollary~\ref{corollary:linearity} (i) and (ii)) means the boundary is governed by a linear function. It implies that any combination of the control system parameters within the region will maintain system stability, offering engineers substantial flexibility in parameter selection and system tuning without compromising stability. This implication is applicable to the convex stability region, where the boundary is governed by a convex function. In contrast, a concave stability region (e.g., Corollary~\ref{corollary:linearity} (iii)) has a boundary governed by a concave function. This indicates that while individual parameter sets within the region ensure stability, linear combinations of these parameters may not. For any stable parameter set, all parameter sets within the rectangular area defined by this set and the origin are also stable. In addition to control system parameters, communication system parameters also impact the stability region, which is presented in Section~\ref{sec:NumSimOfThrem1}.
\begin{figure}[t]
	\centering\includegraphics[width=3.5in]{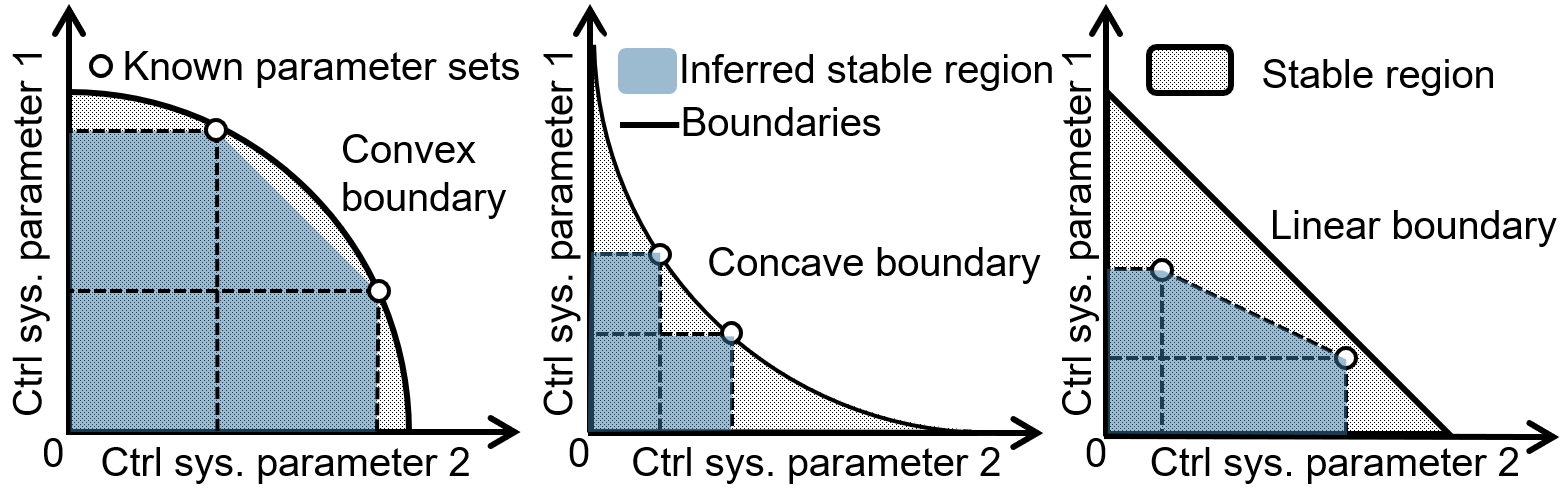}
	\vspace{-0.8cm}
	\caption{Illustration of the stability region boundaries.}
	\label{fig:StabilityRegion}
	\vspace{-0.7cm}
\end{figure}

\subsubsection{Special cases}
Given the stability condition of the general WHMC system in Theorem~\ref{threm:stability}, we examine the stability conditions for three specific cases.

\textbf{For an error-free channel}, assuming the communication channels are perfect, we have $p_H(t)=p_M(t)=0, \forall t$. The stability condition in \eqref{threm:stability} reduces to
\begin{equation} \label{eq:Error-Free}
   \mathbb{E}\left[ \alpha_M^{L} \right]\alpha_{HM} < 1,
\end{equation}
where $\mathbb{E}\left[ \alpha_M^{L} \right] = \sum_{k=1}^{\tau_{\text{max}}} \alpha_M^{k+1}v_k$ and $v_k$ defined in \eqref{eq:DistriHumanState} is determined by the human state transition matrix $\mathbf{M}$. In this case, the stability depends on $\alpha_{M}$, $\alpha_{HM}$, and $\mathbf{M}$. Since the communication channels are perfect, only human control loops may be open due to the human control lag. Thus, only the Lyapunov gains in cases one and two of Assumption~\ref{assump:controlGains}, i.e., $\alpha_{HM}$ and $\alpha_{M}$, play a role in this scenario. 

\textbf{Human control only}, assuming that the plant is only controlled by a human operator, i.e., the machine control loop is always open ($p_M(t)=1, \forall t$). The stability condition is
\begin{equation}\label{SpecialCaseHumanOnly}
   \mathbb{E}\left[ \alpha^{L} \right]\alpha_{H} < 1,
\end{equation}
where $\mathbb{E}\left[ \alpha^{L} \right] = \sum_{l=0}^{\infty} \alpha^{l} z_l$ and $z_l$ is defined in \eqref{eq:MarkovCompelexity}. In this case, the stability depends on $\alpha_{H}$, $\alpha$, and $\mathbf{M}$. Since the machine control loop is always open, only the Lyapunov gains in cases three and four of Assumption~\ref{assump:controlGains} are relevant. We note that if the human control lag is a constant, $L$ is still a random time interval due to the random SH delay.

\textbf{Machine control only}, assuming that the plant is only controlled by a machine controller, i.e., the human control loop is always open ($p_H(t)=1, \forall t$).
The stability condition of this case cannot be directly obtained from Theorem~\ref{threm:stability}, because the stochastic cycle-based approach in Theorem~\ref{threm:stability} is on the basis of closed human control loops. Modifications to the definition of stochastic cycles are required to analyze the stability condition. Our results are presented next: 
\begin{proposition}\label{prop:MachineStability}
    \normalfont
    The plant in Section~\ref{sec:sys} controlled solely by the machine is stochastically stable if
    \begin{equation}\label{SpecialCaseMachineOnly}
   \frac{\alpha_{M}}{\alpha}\mathbb{E}\left[ \alpha^{\hat{L}} \right] < 1,
    \end{equation}
where control system parameters $\alpha$ and $\alpha_{M}$ are defined in Assumption~\ref{assump:controlGains}; $\hat{L}$ is the time steps between the two consecutive closed machine control loops with the probability distribution of $\mathbb{P}\left[\hat{L} = l \right] = (1 - \Bar{p}_M) \Bar{p}_M^{l - 1}$.
\end{proposition}
\begin{proof}
	See Appendix~\ref{proofProp1}.
\end{proof}
In this case, the stability depends on $\alpha$, $\alpha_{M}$ and $\Bar{p}_M$. Since the human control loop is always open in this case, only the Lyapunov gains in cases two and four of Assumption~\ref{assump:controlGains} are applicable. \eqref{SpecialCaseMachineOnly} resemble exist results \cite{LiuStability}.

\subsection{Numerical Examples of the Stability Region} \label{sec:NumSimOfThrem1}
\begin{figure*}[t]
    \centering
    \includegraphics[width=6.0in]{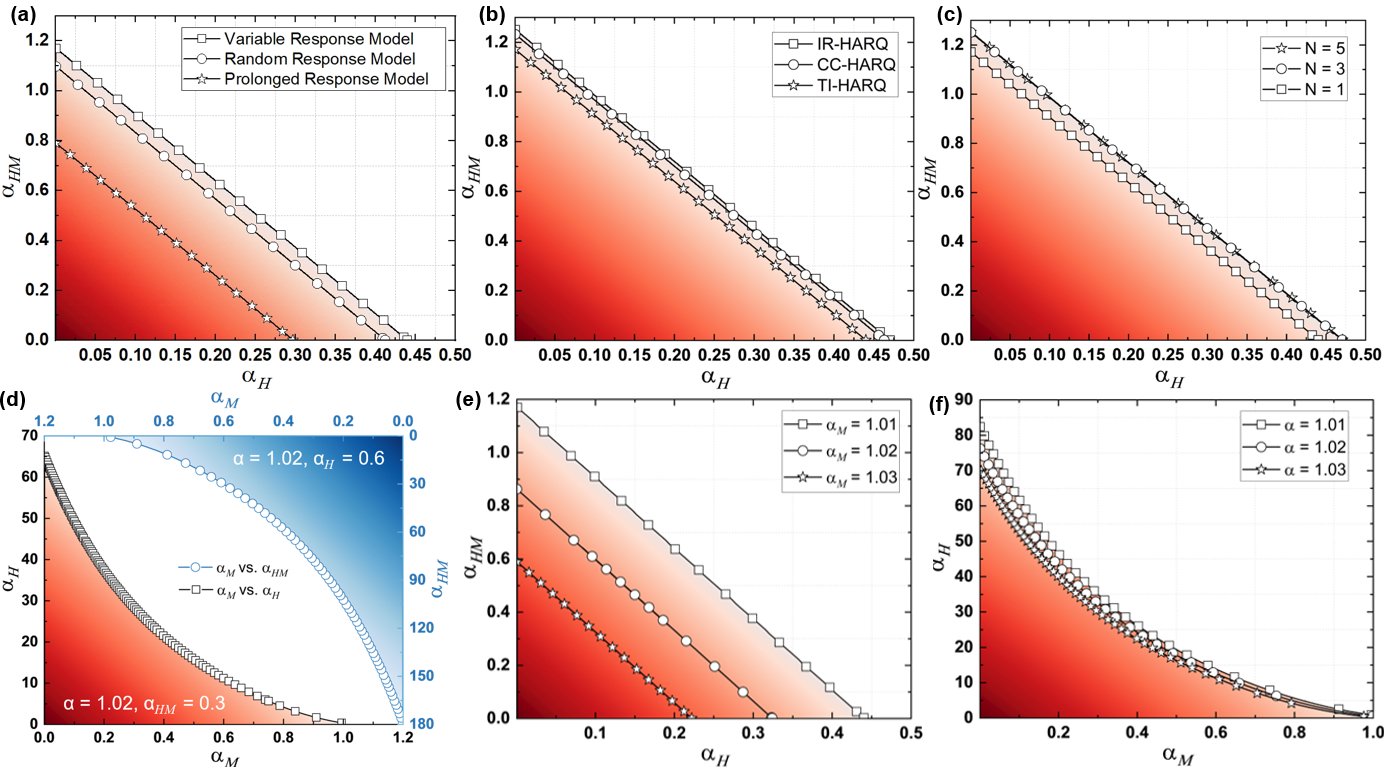}
    \vspace{-0.4cm}
    \caption{Numerical examples of the boundary of stability conditions: (a) Impacts of the human state transition matrix on the stability region in terms of $\alpha_{HM}$ and $\alpha_{H}$, where $\alpha=1.02$, $\alpha_{M}=1.01$, $N=3$, and TI-HARQ are adopted. (b) Impacts of HARQ schemes on the stability region in terms of $\alpha_{HM}$ and $\alpha_{H}$, where $\alpha=1.02$, $\alpha_{M}=1.01$, $\mathbf{M} = \mathbf{M}_l$, and $N=3$ are adopted. (c) Impacts of the maximum number of retransmissions on the stability region in terms of $\alpha_{HM}$ and $\alpha_{H}$, where $\alpha=1.02$, $\alpha_{M}=1.01$, $\mathbf{M} = \mathbf{M}_l$, and IR-HARQ are adopted. (d) The stability region in terms of $\alpha_{M}$ vs. $\alpha_{H}$ and $\alpha_{M}$ vs. $\alpha_{HM}$, where $\mathbf{M} = \mathbf{M}_l$, and IR-HARQ are adopted. (e) The stability region in terms of $\alpha_{HM}$ and $\alpha_{H}$, where $\alpha=1.02$, $\mathbf{M} = \mathbf{M}_l$, and IR-HARQ are adopted. (f) The stability region in terms of $\alpha_{M}$ and $\alpha_{H}$, where $\alpha_{HM}=0.3$, $\mathbf{M} = \mathbf{M}_l$, and IR-HARQ are adopted. Colourized areas are stable regions.} 
    \label{fig:expStabilityCondition}
    \vspace{-0.5cm}
\end{figure*}

We present numerical results to illustrate the stability region in terms of the communication, the control system, and the human model parameters, which show how these parameters affect the stability condition \eqref{eq:Stability} in Theorem~\ref{threm:stability}. The average channel gain is denoted as $\bar{h}$ and follows the free-space path loss model $\bar{h} = A(\frac{3\times10^8}{4\pi f_cd})^{d_e}$, where $A$ denotes the antenna gain; $f_c$ denotes the carrier frequency; $d$ denote the distance from the human operator or the machine to the plant; $d_e$ denote the path loss exponent \cite{PathLossModel}. The time-varying wireless channel power gains are generated from Rayleigh fading channel models, i.e., $h(t) \sim\!Exp(1)$. Given the transmission power $P_{\text{tx}}$ and the receiving noise power $\sigma^2$, the SNR of received packets in all channels are obtained from $\gamma(t) = \frac{\bar{h}h(t) P_{\text{tx}}}{\sigma^2}$, respectively. The communication parameters are summarized in Table~\ref{Para:NExp}. 
\begin{table}[!t] 
\footnotesize
\setlength\tabcolsep{13pt}
\centering
\caption{Communication Parameters in Simulation}
\vspace{-0.3cm}
\begin{tabular}{cc}
\hline\hline
\textbf{Items}                    & \textbf{Value} \\ \hline
\multicolumn{2}{l}{\textbf{Communication parameters}}  \\
\rowcolor[HTML]{EFEFEF} 
Code rate {[}bps{]}, $b/l_p$      &   2               \\
Packet length {[}symbols{]}, $l_p$ &      1500            \\
\rowcolor[HTML]{EFEFEF} 
Transmit power {[}dBm{]}, $P_\text{tx}$  &    23          \\
Background noise power  {[}dBm{]}, $\sigma^2$ &   -70         \\
\rowcolor[HTML]{EFEFEF} 
Maximum number of re/transmissions, $N$ &    $\{1, 3, 5\}$   \\
\multicolumn{2}{l}{\textbf{Free-space path loss model}} \\
\rowcolor[HTML]{EFEFEF} 
Antenna gain, $A$  &     4 \\
Carrier frequency [MHz], $f_c$   &  915 \\
\rowcolor[HTML]{EFEFEF} 
Distance from machine to plant [m], $d$   &   40 \\
Distance from human to plant [m], $d$  &        45 \\
\rowcolor[HTML]{EFEFEF} 
Path loss exponent, $d_e$   & 2.9 \\\hline\hline
\end{tabular}
\label{Para:NExp}
\vspace{-0.7cm}
\end{table}

The human control lag has two states $\mathcal{S} = \{ 5, 25\}$ (i.e., fast and slow) with the stationary probability distribution $(0.5,0.5)$ and the state transition matrix $\mathbf{M}$ can be one of the three cases below:
\begin{equation}
    \mathbf{M}_h\!=\!\begin{bmatrix} 0.9 & 0.1 \\ 0.1 & 0.9\end{bmatrix}, \mathbf{M}_e\!=\!\begin{bmatrix} 0.5 & 0.5 \\ 0.5 & 0.5\end{bmatrix}, \mathbf{M}_l\!=\!\begin{bmatrix} 0.1 & 0.9 \\ 0.9 & 0.1\end{bmatrix}.
\end{equation}
\(\mathbf{M}_h\) is a Prolonged Response Model, where the human operator tends to remain in a single state—either fast (low lag) or slow (high lag)—for extended periods. This reflects a tendency for the operator's reaction time to be consistently fast or slow, with infrequent transitions between these two states.
\(\mathbf{M}_e\) is a Random Response Model, where the human operator has an equal probability of staying in their current state or switching to the other, leading to unpredictable shifts between fast and slow reactions.
\(\mathbf{M}_l\) is a Variable Response Model, where the human operator frequently switches between fast and slow reactions, indicating high variability in response times.

Numerical results are illustrated in Fig.~\ref{fig:expStabilityCondition}. We select the pair of $\alpha_H$ and $\alpha_{HM}$ to show the impacts because this pair has the simplest linear relationship for demonstration (see Corollary~\ref{corollary:linearity}). Fig.~\ref{fig:expStabilityCondition}(a) illustrates the impacts of human model parameters on the stability region. In particular, a human operator with a variable response model shows the largest stability region, while a human operator with a prolonged response model has the smallest stability region. A human operator with a prolonged response model has a higher chance of instantly staying in a large lag state. Thus, to guarantee closed loop stability, more reliable communications are required. As shown in Corollary~\ref{corollary:linearity}(i), the slope of the linear stability region in Fig.~\ref{fig:expStabilityCondition}(a) depends on the expected probability of an open machine control loop $\Bar{p}_M$ defined in \eqref{ExPr_OpenMachineLoop}.

Fig.~\ref{fig:expStabilityCondition}(b) illustrates the impacts of three HARQ schemes on the stability region. Compared with TI-HARQ, WHMC systems with IR-HARQ and CC-HARQ schemes show a larger stability region due to the fact that the packet combining can significantly reduce the number of retransmissions by taking advantage of the accumulated SNRs. A WHMC system with the IR-HARQ scheme has the largest stability region because only incremental redundancies are retransmitted for each event of the erroneous packet.
Fig.~\ref{fig:expStabilityCondition}(c) illustrates the impacts of maximum re/transmission attempts on the stability region. We see that the system with HARQ schemes (i.e., $N > 1$) has a larger stability region than the system without retransmission (i.e., $N = 1$). As $N$ increases, the extension of the stability region becomes small; thus, $N = 3$ is commonly used in the numerical illustrations. 

In addition to the linear boundary, Fig.~\ref{fig:expStabilityCondition}(d) illustrates the concave boundaries in terms of $\alpha_{M}$ vs. $\alpha_{H}$ and $\alpha_{M}$ vs. $\alpha_{HM}$, where a small variation of $\alpha_{M}$ leads to a significant change in both $\alpha_{H}$ and $\alpha_{HM}$. This is because machine control attempts are more frequent than human control attempts, and the accumulated significance of $\alpha_{M}$ is significantly higher than the Lyapunov gains $\alpha_{H}$ and $\alpha_{HM}$ involving human control attempts.
Fig.~\ref{fig:expStabilityCondition}(e) illustrates stability regions in terms of the pair of $\alpha_{HM}$ and $\alpha_{H}$ with different $\alpha_{M}$. As $\alpha_{M}$ decreases, the stability region expands dramatically, highlighting the significant reduction of human control efforts to stabilize the system.
Fig.~\ref{fig:expStabilityCondition}(f) illustrates stability regions in terms of the pair of $\alpha_{M}$ and $\alpha_{H}$ with different $\alpha$. The stability region expands with decreasing $\alpha$. This is because a larger open loop Lyapunov gain $\alpha$ indicates greater effort required for both automatic machine and human control inputs.

\section{Proof of Concept Experiment} \label{sec:simulation}
In this section, we present a case study of WHMC to illustrate its advantage in control performance. The experiment data of the case study are recorded to estimate the control system and the
human model parameters, followed by the stability analysis of the case study to show the effectiveness of Theorem~\ref{threm:stability}.

\begin{figure}
    \centering
    \includegraphics[width=2.9in]{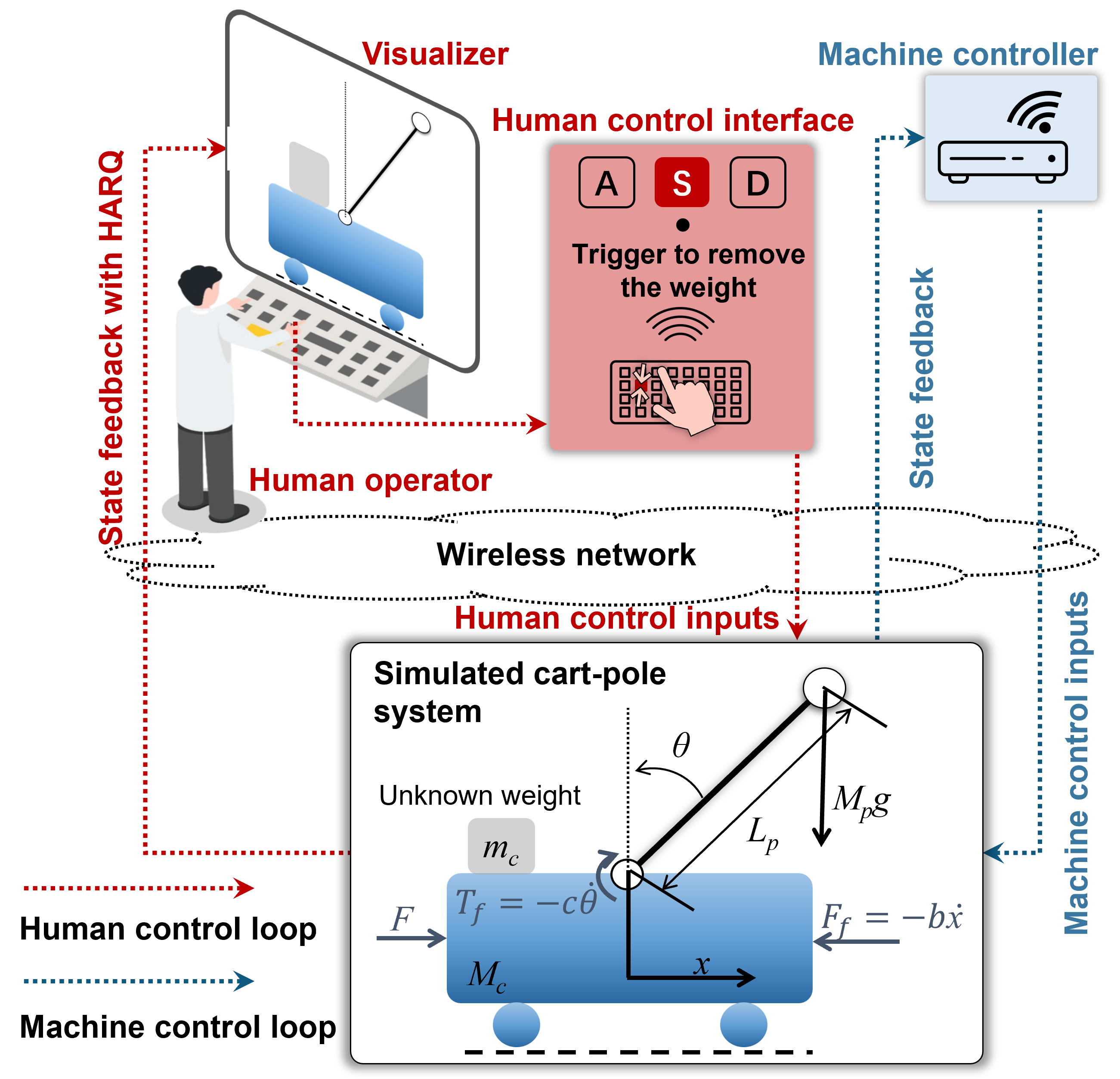}
    \vspace{-0.4cm}
    \caption{The cart-pole system to illustrate the WHMC.}
    \label{fig:CartPole}
    \vspace{-0.7cm}
\end{figure}

\subsection{Experiment Setups}\label{CaseStudy:Setup}
We build a WHMC system where a cart-pole system is simulated and controlled by a machine controller and a real human operator, as shown in Fig.~\ref{fig:CartPole}. The machine controller is implemented to control the applied force $F$ to the cart with an unknown weight $m_c$ for balancing the pole. The dynamic weight $m_c$ on the cart can be observed by the human operator who monitors the state of the simulated cart-pole system and uses a key `S' between `A' and `D' on a keyboard to intervene in the control of the cart-pole system to remove the dynamic weight on the cart. The dynamic weight can be seen as a catastrophic disturbance to the system, which cannot be handled by the machine controller designed without such knowledge. Therefore, such a control system has nonlinear dynamics and unknown disturbance to the machine controller, which is challenging without collaboration with a human operator. The IR-HARQ scheme is adopted with a maximum re/transmissions number of $N = 3$. Other communication parameters and the free-space path loss model are the same as Table~\ref{Para:NExp}.

\subsubsection{Cart-pole dynamics}
In the simulated cart-pole system, the mass of the pole is assumed to be concentrated at its end mass. The states of the cart-pole system consist of the position $x(t)$ and velocity $\dot{x}(t)$ of the cart, the angle $\theta(t)$ and angular velocity $\dot{\theta}(t)$ of the pole, and the unknown weight $m_c(t)$ on the cart, which is denoted as $\mathbf{x}(t) \triangleq (x(t),\dot{x}(t),\theta(t),\dot{\theta}(t),m_c(t))^\top$. The dynamics of the cart-pole system are governed by the non-linear dynamic equations in \eqref{eq:plantDynamics},
\begin{figure*}[!t]
\normalsize
\vspace{-0.3cm}
%\hrulefill
\begin{equation}\label{eq:plantDynamics}
    \begin{cases}
        2M_pgL_p\sin{\!\theta(t)}\!=\!(2I+M_pL_p^{2})\ddot{\theta}(t)\!+\!M_pL_p\cos{\!\theta(t)}\ddot{x}(t)\!+\!2c\dot{\theta}(t),\!&\text{Rotational dynamics (pendulum)},\\
        M_pL_p\sin{\!\theta(t)}\dot{\theta}(t)^2\!=\!(M_c\!+\!M_p\!+\!m_c(t))\ddot{x}(t)\!+\!b\dot{x}(t)\!+\!M_pL_p\cos{\!\theta(t)}\ddot{\!\theta}(t)\!-\!u_M(t),\!&\text{Horizontal force balance (cart)},
    \end{cases}
\end{equation}
%\hrulefill
\vspace{-0.5cm}
\end{figure*}
where $M_p = 2$ $kg$ and $M_c= 10$ $kg$ are the mass of the pole and cart, respectively; $g= 9.8$ $m/s^{2}$ is the gravitational acceleration; $L_p = 6$ $m$ is the length of the pole; $I = \frac{M_pL_p^{2}}{4}$ is the moment of inertia for a point mass in terms of the center of the pole; $u_M(t)$ is the applied force to the cart by the machine controller in \eqref{eq:MachineControl}; $c = b = 0.1$ are the damping coefficients for the pole and cart, respectively. 
For the dynamics of the unknown weight on the cart, $m_c(t)$, we assume that once the weight is successfully removed by the actuator remotely controlled by the human operator, it will reappear on the cart after a random time interval; otherwise, the unknown weight will remain on the cart continuously. Thus, $m_c(t)$ has the following updating rule
\begin{equation}\label{eq:NoiseUpdate}
    m_c(t+1) = \begin{cases}
        m_c(t) + u_H(t), & \text{for } m_c(t) \neq 0,\\
        m_c(t) + w(t), & \text{otherwise,}
    \end{cases}
\end{equation}
where $w(t) \in \{0, 5\}$ is randomly generated and $u_H(t)$ is the human control input. 
The sampling period is $T_s = 0.05$ $s$. $\ddot{x}(t)$ and $\ddot{\theta}(t)$ can be derived from \eqref{eq:plantDynamics} given $\mathbf{x}(t)$. Then, by leveraging $\mathbf{x}(t)$, $\ddot{x}(t)$, $\ddot{\theta}(t)$, $T_s$ and \eqref{eq:NoiseUpdate}, we can get $\mathbf{x}(t+1)$, indicating the proposed cart-pole system follows \eqref{NLTI}. The initial state is $\mathbf{x}(0) = (0,0,\frac{\pi}{6},0,5)^\top$.

\subsubsection{Control policies}
In this experiment, $m_c(t)$ is unknown to the machine controller, but all other states, parameters, and the system dynamics in \eqref{eq:plantDynamics} (excluding $m_c(t)$) are known. The machine control policy seeks to achieve a decaying angle as per $\theta(t+1) = \eta\theta(t)$ by applying force $u_M(t)$, where $\eta \in (0, 1)$.
Using the Euler approximation to update the angle, we obtain  
\begin{equation} \label{eq:ExpMachineControlEffects}
\theta(t+1) = \dot{\theta}(t +1)T_s + \theta(t) = \eta\theta(t).
\end{equation}
The updated angular velocity $\dot{\theta}(t+1)$ is also obtained by using the Euler approximation, i.e., $\dot{\theta}(t+1) = (\dot{\theta}(t)+T_s\ddot{\theta}(t))T_s$. A smaller $\eta$ will lead to a control policy enabling a faster speed of $\theta(t) \rightarrow 0$, which is 0.7 in the experiment.
By leveraging \eqref{eq:ExpMachineControlEffects} and \eqref{eq:plantDynamics} with $m_c(t) = 0$, the machine control policy can be obtained as \eqref{eq:ExpControl} to determine $\ddot{\theta}(t)$,
\begin{figure*}[!t]
\normalsize
\vspace{-0.2cm}
\begin{equation} \label{eq:ExpControl}
\begin{aligned}
    f_M \left(\mathbf{x}(t) \right) = &\frac{2M_pgL_p\sin(\theta(t))(M_c+M_p)}{M_pL_p\cos(\theta(t))}-\frac{2c\dot{\theta}(t)(M_c+M_p)}{M_pL_p\cos(\theta(t))}+b\dot{x}(t)-M_pL_p\dot{\theta}^2\sin(\theta(t))\\&
    -\frac{(\eta-1)\theta(t)\Gamma(t)}{T_s^2M_pL_p\cos(\theta(t))} + \frac{\dot{\theta}(t)\Gamma(t)}{T_sM_pL_p\cos(\theta(t))},
\end{aligned}
\end{equation}
\vspace{-0.6cm}
\hrulefill
\end{figure*}
where $\Gamma(t) \triangleq (M_c+M_p)(2I+M_pL_p^2)-(M_pL_p\cos(\theta(t)))^2$.
Recall that the human control policy is to remove the unknown weight on the cart if the human operator observes it through visual feedback. Thus, the human control policy is $f_{H}(\mathbf{x}(t)) = - m_c(t)$.

\subsection{WHMC Control Performance}
\begin{definition}[Collaborative Control Performance]
    \normalfont
   The control performance of a WHMC system at each time step is evaluated by a cost function $J$: $\mathbb{R}^{l_s}\rightarrow\mathbb{R}_{\geq0}$, which is defined as 
   \begin{equation}
       J(t) = \mathbf{x}(t)^\top \mathbf{P} \mathbf{x}(t),
   \end{equation}
   where $\mathbf{P}\in\mathbb{R}^{l_s \times l_s}$ is a positive diagonal matrix to individually penalize the states of interest. A smaller control cost indicates a better control performance.
\end{definition}

In the experiment described in Section~\ref{CaseStudy:Setup}, the objective of the WHMC system is to balance the pole (i.e., $\theta(t)$ is closely around the zero point). Thus, we are only interested in the angle of the pole, resulting in a cost function $J(t) = (\theta(t))^2$. The control cost of the machine control only case, the human control only case, and the WHMC case are shown in Fig.~\ref{fig:cost}. The human operator's objective is to remove the weight, not to balance the pole. Only the machine controller handles pole balancing. Without machine control inputs, the control cost increases. Both the WHMC and machine-only cases can reduce the cost over time, with their stability guaranteed, as will be further discussed in Section~\ref{CaseStudy:StabilityCheck}. Compared to the machine-only case, the WHMC case shows a faster decrease in cost, demonstrating the importance of WHMC.
\begin{figure}
    \centering
    \includegraphics[width=3.2in]{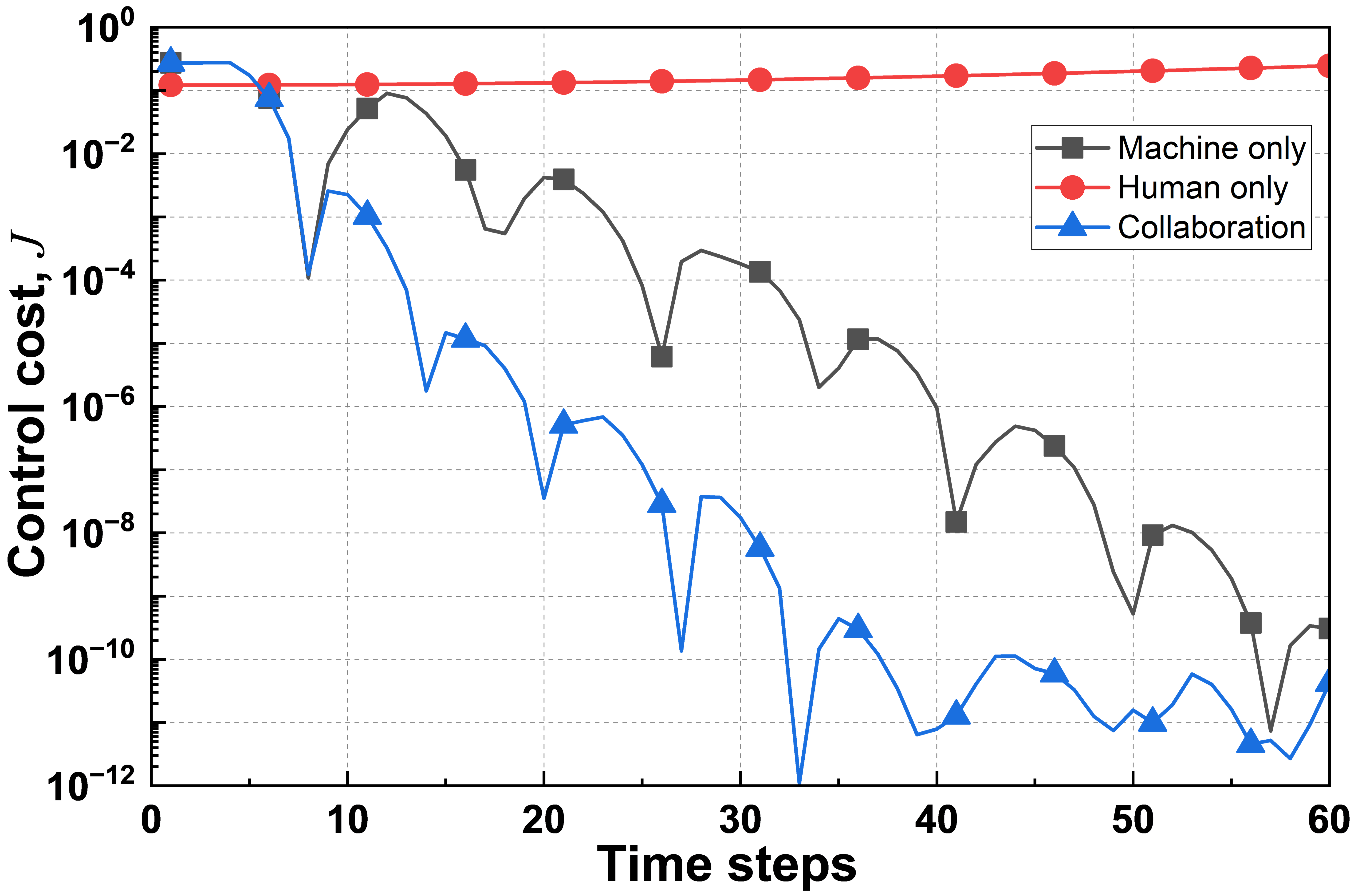}
    \vspace{-0.4cm}
    \caption{The control cost of the cart-pole system.}
    \label{fig:cost}
    \vspace{-0.5cm}
\end{figure}

\subsection{WHMC System Stability}\label{EstimationForStabilityAnalysis}
\subsubsection{Estimation of control system parameters}
Since the control objective is to balance the pole, the Lyapunov-like function $V(\cdot)$ is defined as
\begin{equation} \label{eq:DefinedV}
    V\left(\theta(t)\right) = \begin{cases}
        |\theta(t)|, & \text{for } |\theta(t)|\geq0.05,\\
        0, & \text{otherwise,}
    \end{cases}
\end{equation}
where the threshold of 0.05 is set to eliminate the impacts of uncontrolled $\dot{\theta}(t)$ on the control objective $\theta(t)$ when $\theta(t)$ reaches to the desired zero point.
To estimate the four control system parameters, i.e., $\alpha_{HM}$, $\alpha_{M}$, $\alpha_{H}$, and $\alpha$, we collect data by conducting the experiment in four cases, i.e., no control inputs, machine control only, human control only, and human-machine collaborative control, respectively.  
The parameter in each case is estimated by $\max{\frac{V(\theta(t+1))}{V(\theta(t))}}$ based on the corresponding data set,
where $V(\cdot)$ is defined in \eqref{eq:DefinedV}. The estimated four control system parameters are  $\alpha_{HM}= 0.5271$, $\alpha_{M}=0.7949$, $\alpha_{H}=1.0196$, and $\alpha = 1.0134$.

\subsubsection{Estimation of human model}
To reduce the estimation complexity, we quantize the human control lag to a two-state set of $\{3,7\}$, which corresponds to 0.15 $s$ and 0.35 $s$. The state transition matrix is estimated based on the maximum likelihood estimation approach, which is  
\begin{equation}
    \mathbf{M} = \begin{bmatrix} 
    0.2576 & 0.7424\\
    0.4404 & 0.5596
    \end{bmatrix}.
\end{equation}
The corresponding stationary probability distribution is $(0.3723, 0.6277)$, i.e., $\mathbb{P}\!\left[ \tau_{H}(t')\!=\!0.15 \right]\!=\!0.3723$ and $\mathbb{P}\!\left[ \tau_{H}(t')\!=\!0.35 \right]\!=\!0.6277$.

\subsubsection{Stability of the cart-pole system} \label{CaseStudy:StabilityCheck}
Based on the above estimation and parameters in Table~\ref{Para:NExp}, the left term of the stability condition in Theorem~\ref{threm:stability} is $0.3539<1$, demonstrating a stabilized WHMC system. 
In the machine-only control scenario, the left term of the stability condition \eqref{SpecialCaseMachineOnly} is $0.8594<1$, indicating a stochastically stable system. Conversely, in the human-control-only case, the left term of the stability condition \eqref{SpecialCaseHumanOnly} is $3.3088>1$, signifying an unstable system. This instability also explains the increasing control cost observed in Fig.~\ref{fig:cost}.

\section{Conclusions}\label{sec:conclusion}
We have developed a foundational WHMC model that integrates dual wireless loops for both machine and human control, addressing the intricate challenges associated with WHMC systems. By introducing a novel stochastic cycle-cost-based approach, we have derived a stability condition that accounts for the complexities of wireless communication, human behavior, and control system dynamics. Our approach has been validated through extensive numerical analysis and the creation of a new case study, demonstrating its practical effectiveness. These contributions offer a strong basis for advancing WHMC systems in increasingly complex and dynamic environments.
%In future work, we will explore optimal communication design and radio resource allocation strategies that ensure stability in large-scale WHMC systems involving multiple plants and human operators, paving the way for more robust and scalable implementations in real-world applications.

\section*{Acknowledgments}
The authors would like to express their sincere gratitude to Dr. Anuradha Annaswamy, Director of the Active-Adaptive Control Laboratory at MIT, for her valuable comments on this paper. Her insightful feedback and suggestions have been instrumental in improving the clarity and rigor of this work.

\begin{appendices}
\renewcommand{\thesectiondis}[2]{\Alph{section}:}
\section{Proof of Theorem~\ref{threm:stability}} \label{proofThrem1}
The time steps of the $n$th closed-loop human control is defined as $t = k_n$, as shown in Fig.~\ref{fig:CycleBasedAnalysis.png}. Let $m_C$ and $m_O$ denote the number of case two and case four defined in Assumption~\ref{assump:controlGains} between $t = k_n$ and $ t = k_{n+1}$, respectively.
Then we have 
\begin{equation}
   V(\mathbf{x}(k_n+l)) \leq \alpha_M^{m_C}\alpha^{m_O}V(\mathbf{x}(k_n)),
\end{equation}
and 
\begin{equation}
   \mathbb{E} \left[ V(\mathbf{x}(k_n+l)) \right] \leq \mathbb{E} \left[\alpha_M^{m_C}\alpha^{m_O} \right] \mathbb{E} \left[ V(\mathbf{x}(k_n))\right].
\end{equation}
Since
\begin{equation}
    \mathbb{E} \left[\alpha_M^{m_C}\alpha^{m_O} \mid m_C+m_O = l \right] =  \left(\alpha_M \left( 1-\Bar{p}_M\right) + \alpha \Bar{p}_M \right)^l,
\end{equation}
the sum of $V(\cdot)$ between the two adjacent closed human control loops has the following inequality
\begin{equation}
    \sum_{t = k_n}^{k_{n+1}-1} \mathbb{E} \left[ V(\mathbf{x}(t)) \right] \leq \sum_{l = 0}^{k_{n+1}-k_{n}-1} \Omega^l \mathbb{E} \left[ V(\mathbf{x}(k_n))\right],
\end{equation}
where 
$$\Omega \triangleq \alpha_M \left( 1-\Bar{p}_M\right) + \alpha \Bar{p}_M.$$
By further processing the above inequality, we have
\begin{equation}
    \sum_{t = k_n}^{k_{n+1}-1} \mathbb{E} \left[ V(\mathbf{x}(t)) \right] \leq \frac{1-\Omega^{L-1}}{1-\Omega} \mathbb{E} \left[ V(\mathbf{x}(k_n))\right].
\end{equation}
Then, 
\begin{equation} \label{eq:infSum}
    \sum_{t = 0}^{\infty} \mathbb{E} \left[ V(\mathbf{x}(t)) \right] \leq \sum_{n=0}^{\infty} \frac{1-\Omega^{L-1}}{1-\Omega} \mathbb{E} \left[ V(\mathbf{x}(k_n))\right].
\end{equation}
Let $h_C$ and $h_O$ denote the numbers of case one and case three defined in Assumption~\ref{assump:controlGains} between $t = 0$ and $ t = k_{n}$, respectively. In this time interval, $\hat{m}_C$ and $\hat{m}_O$ denote the numbers of case two and case four, respectively. Then, 
\begin{equation} \label{eq:InitialToKn}
   \mathbb{E} \left[ V(\mathbf{x}(k_n)) \right] \leq \mathbb{E} \left[\alpha_M^{\hat{m}_C} \alpha^{\hat{m}_O} \alpha_{HM}^{h_C} \alpha_{H}^{h_O} \right] \mathbb{E} \left[ V(\mathbf{x}(0))\right].
\end{equation}
It can be further processed as
\begin{equation}
   \mathbb{E} \left[ V(\mathbf{x}(k_n)) \right] \leq \mathbb{E} \left[\Omega^{nL} \right] \mathbb{E} \left[ \alpha_{HM}^{h_C} \alpha_{H}^{h_O} \right] \mathbb{E} \left[ V(\mathbf{x}(0))\right].
\end{equation}
Since
\begin{equation}
  \mathbb{E} \left[ \alpha_{HM}^{h_C} \alpha_{H}^{h_O} \mid h_C + h_O = n\right]  = \left(\alpha_{HM} \left( 1-\Bar{p}_M\right) + \alpha_H \Bar{p}_M \right)^{n},
\end{equation}
we have
\begin{equation} \label{eq:initialSum}
   \mathbb{E} \left[ V(\mathbf{x}(k_n)) \right] \leq \mathbb{E} \left[\Omega^{nL} \right] \Lambda^{n} \mathbb{E} \left[ V(\mathbf{x}(0))\right],
\end{equation}
where 
$$\Lambda \triangleq \alpha_{HM} \left( 1-\Bar{p}_M\right) + \alpha_H \Bar{p}_M.$$
By leveraging \eqref{eq:infSum} and \eqref{eq:initialSum}, we have
\begin{equation}
    \sum_{t = 0}^{\infty} \mathbb{E} \left[ V(\mathbf{x}(t)) \right] \leq \sum_{n=0}^{\infty} \frac{1-\Omega^{L-1}}{1-\Omega} \mathbb{E} \left[\Omega^{nL} \right] \Lambda^{n} \mathbb{E} \left[ V(\mathbf{x}(0))\right].
\end{equation}
Since $\mathbb{E}\left[V (\mathbf{x}(0))\right] < \infty$, to make $\sum_{t=0}^{\infty} \mathbb{E}\left[V\left(\mathbf{x}(t)\right)\right] < \infty$, we need
\begin{equation} \label{eq:Reqired}
   \sum_{n=0}^{\infty} \frac{1-\Omega^{L-1}}{1-\Omega} \mathbb{E} \left[\Omega^{nL} \right] \Lambda^{n} < \infty.
\end{equation}
Let 
\begin{equation}
   \mathbb{E}\left[\Xi (n)\right] \triangleq \mathbb{E}\left[\frac{1-\Omega^{L-1}}{1-\Omega} \Omega^{nL} \right]  \Lambda^{n}.
\end{equation}
Then, we have
\begin{equation}
   \mathbb{E}\left[\Xi (n+1)\right] \triangleq \mathbb{E}\left[\frac{1-\Omega^{L-1}}{1-\Omega} \Omega^{(n+1)L} \right] \Lambda^{n+1},
\end{equation}
and
\begin{equation} \label{eq:Xi}
   \mathbb{E}\left[\Xi (n+1)\right] = \mathbb{E}\left[ \Omega^{L} \right] \Lambda \mathbb{E}\left[\Xi (n)\right].
\end{equation}
To satisfy \eqref{eq:Reqired}, \eqref{eq:Stability} is derived from \eqref{eq:Xi} as the stability condition of the WHMC system.

\section{Proof of Corollary~\ref{corollary:linearity}} \label{proofCoro1}
According to \eqref{eq:Stability}, the bound of the stability region is 
    \begin{equation} \label{eq:linearEqual0}
        \mathbb{E}\left[ \Omega^{L} \right]\Lambda = 1.
    \end{equation}
    
    (i) When control system parameters $\alpha_M$ and $\alpha$ are fixed, by further processing \eqref{eq:linearEqual0}, we have
    \begin{equation}
        \alpha_{HM}  = \frac{1}{\mathbb{E}\left[ \Omega^{L} \right] \left( 1-\Bar{p}_M\right)} - \frac{\Bar{p}_M }{\left( 1-\Bar{p}_M\right)}\alpha_H,
    \end{equation}
    where the linearity between $\alpha_{HM}$ and $\alpha_H$ is showcased.

    (ii) When control system parameters $\alpha_{HM}$ and $\alpha_H$ are fixed, by further processing \eqref{eq:linearEqual0}, we have
    \begin{equation}\label{linearEquations}
         \sum_{l = 1}^{\bar{L}} \Omega^{l}\mathbb{P}\left[L = l\right] = \sum_{l = 1}^{\bar{L}}\frac{1}{\bar{L}\Lambda},
    \end{equation}
    which is a sum of $\bar{L}\triangleq \mathbb{E}\left[L\right]$ linear equations and can be represented as
    \begin{equation}
        \alpha_M  = - \frac{\Bar{p}_M}{1-\Bar{p}_M}\alpha   + \frac{1}{\left( 1-\Bar{p}_M\right)\bar{L}}\sum_{l = 1}^{\bar{L}}\left( \bar{L} \Lambda \mathbb{P}\left[L = l\right]\right)^{-l}.
    \end{equation}
    Thus, the linearity between $\alpha_M$ and $\alpha$ is demonstrated.
    
    (iii) For any other possible pairs of two control system parameters, we can take one of $\bar{L}$ equations from \eqref{linearEquations} and prove that it is convex in terms of the pairs other than those in (i) and (ii). In particular, we have
    \begin{equation}
        \Omega^{l}\mathbb{P}\left[L = l\right] = \frac{1}{\bar{L}\Lambda}.
    \end{equation}
    Then, it can be represented as
    \begin{equation} \label{fullLinearEq}
        (\alpha_M \left( 1-\Bar{p}_M\right) + \alpha \Bar{p}_M)^{l} = \frac{\frac{1}{\bar{L}\mathbb{P}\left[L = l\right]}}{(\alpha_{HM} \left( 1-\Bar{p}_M\right) + \alpha_H \Bar{p}_M)}.
    \end{equation}
    We take the pair of $\alpha$ and $\alpha_{H}$ for example, given the fixed $\alpha_{M}$ and $\alpha_{HM}$. Then, \eqref{fullLinearEq} can be represented as
    \begin{equation} \label{oneLinearEq}
        \alpha\!=\!\frac{1}{\Bar{p}_M}\!\!\left(\!\!\left(\!\!\frac{\frac{1}{\bar{L}\mathbb{P}\left[L = l\right]}}{(\alpha_{HM}\!\left( 1\!-\!\Bar{p}_M\right) + \alpha_H \Bar{p}_M)}\!\!\right)^{\!\!\frac{1}{l}}\!\!-\alpha_M \!\left(1\!-\!\Bar{p}_M\right)\!\!\right)\!\!.
    \end{equation}
    The first-order derivative $\dot{\alpha}(\alpha_{H})$ is 
    \begin{equation}
        \dot{\alpha}(\alpha_{H}) =-\frac{(\alpha_{HM} \left( 1-\Bar{p}_M\right) + \alpha_H \Bar{p}_M)^{-1-\frac{1}{l}}}{l(\bar{L}\mathbb{P}\left[L = l\right])^{\frac{1}{l}}}.
    \end{equation}
    The second-order derivative $\ddot{\alpha}(\alpha_{H})$ is
    \begin{equation}
        \ddot{\alpha}(\alpha_{H}) = \frac{(1+\frac{1}{l})\Bar{p}_M(\alpha_{HM} \left( 1-\Bar{p}_M\right) + \alpha_H \Bar{p}_M)^{-2-\frac{1}{l}}}{l(\bar{L}\mathbb{P}\left[L = l\right])^{\frac{1}{l}}}.
    \end{equation}
    We note that $\Bar{p}_M \in (0,1)$, $\alpha_{H} \geq 0$ and $\alpha_{HM} \geq 0$. Thus, $\ddot{\alpha}(\alpha_{H}) \geq 0$. Then, \eqref{oneLinearEq} is convex and the sum of convex functions \eqref{linearEquations} is convex and has a concave stability boundary. Other pairs other than those in (i) and (ii) can also be proved following the above analysis.

\section{Proof of Proposition~\ref{prop:MachineStability}} \label{proofProp1}
We also leverage the stochastic cycle-based approach in Section~\ref{sec:StochasticAnalysis}. Assume the time steps of the two adjacent closed machine control loops are $k_n$ and $k_{n+1}$. Then, we have 
 \begin{equation}
   V(\mathbf{x}(k_n+l)) \leq \alpha^{l}V(\mathbf{x}(k_n)),
\end{equation}
and 
\begin{equation}
   \mathbb{E} \left[ V(\mathbf{x}(k_n+l)) \right] \leq \alpha^{l} \mathbb{E} \left[ V(\mathbf{x}(k_n))\right].
\end{equation}
The sum of $V(\cdot)$ between the two adjacent closed machine control loops has the following inequality
\begin{equation}
    \sum_{t = k_n}^{k_{n+1}-1} \mathbb{E} \left[ V(\mathbf{x}(t)) \right] \leq \sum_{l = 0}^{k_{n+1}-k_{n}-1} \alpha^{l} \mathbb{E} \left[ V(\mathbf{x}(k_n))\right].
\end{equation}
By further processing the above inequality, we have
\begin{equation}
    \sum_{t = k_n}^{k_{n+1}-1} \mathbb{E} \left[ V(\mathbf{x}(t)) \right] \leq \frac{1-\alpha^{\hat{L}-1}}{1-\alpha} \mathbb{E} \left[ V(\mathbf{x}(k_n))\right].
\end{equation}
Then, 
\begin{equation}
    \sum_{t = 0}^{\infty} \mathbb{E} \left[ V(\mathbf{x}(t)) \right] \leq \sum_{n=0}^{\infty} \frac{1-\alpha^{\hat{L}-1}}{1-\alpha} \mathbb{E} \left[ V(\mathbf{x}(k_n))\right].
\end{equation}
Since 
\begin{equation} 
   \mathbb{E} \left[ V(\mathbf{x}(k_n)) \right] \leq \alpha_M^{n}\mathbb{E} \left[\alpha^{n\left(\hat{L} - 1 \right)}\right] \mathbb{E} \left[ V(\mathbf{x}(0))\right],
\end{equation}
we have
\begin{equation}
    \sum_{t = 0}^{\infty} \mathbb{E} \left[ V(\mathbf{x}(t)) \right] \leq \sum_{n=0}^{\infty} \frac{1-\alpha^{\hat{L}-1}}{1-\alpha}  \alpha_M^{n}\mathbb{E} \left[\alpha^{n \left(\hat{L} - 1 \right)}\right] \mathbb{E} \left[ V(\mathbf{x}(0))\right].
\end{equation}
Since $\mathbb{E}\left[V (\mathbf{x}(0))\right] < \infty$, to make $\sum_{t=0}^{\infty} \mathbb{E}\left[V\left(\mathbf{x}(t)\right)\right] < \infty$, we need
\begin{equation} \label{eq:Reqired2}
     \sum_{n=0}^{\infty} \frac{1-\alpha^{\hat{L}-1}}{1-\alpha}  \alpha_M^{n}\mathbb{E} \left[\alpha^{n \left(\hat{L} - 1 \right)}\right] < \infty.
\end{equation}
Let
\begin{equation}
   \mathbb{E}\left[\Xi (n)\right] \triangleq \mathbb{E}\left[\frac{1-\alpha^{\hat{L}-1}}{1-\alpha} \alpha^{n \left(\hat{L} - 1 \right)} \right] \alpha_M^{n}.
\end{equation}
Then we have
\begin{equation}
   \mathbb{E}\left[\Xi (n+1)\right] = \mathbb{E}\left[\frac{1-\alpha^{\hat{L}-1}}{1-\alpha} \alpha^{(n+1) \left(\hat{L} - 1 \right)} \right] \alpha_M^{n+1}.
\end{equation}
Then we have the following equation
\begin{equation} \label{eq:Xi2}
\mathbb{E}\!\left[\Xi (n\!+\!1)\right]\!=\!\alpha_M \mathbb{E}\left[\alpha^{\hat{L}-1} \right]\!\!\mathbb{E}\!\left[\Xi (n)\right]\!=\!\frac{\alpha_{M}}{\alpha}\mathbb{E}\!\left[ \alpha^{\hat{L}} \right]\!\!\mathbb{E}\!\left[\Xi (n)\right].
\end{equation}
The stability condition in \eqref{SpecialCaseMachineOnly} is derived from \eqref{eq:Xi2} to satisfy \eqref{eq:Reqired2}.
\end{appendices}

    \balance
	\ifCLASSOPTIONcaptionsoff
	\newpage
	\fi

	\bibliographystyle{IEEEtran}

\end{document}